\documentclass[sigconf]{acmart}

\usepackage{bm}
\usepackage{mathtools}
\usepackage{graphicx}

\title{Preservation Theorems for Transducer Outputs}

\author{Val\'erie Berth\'e}
\affiliation{%
  \institution{Universit\'e Paris Cit\'e, IRIF, CNRS}
\country{France}
}

\author{Herman Goulet-Ouellet}
\affiliation{%
  \institution{Moncton University}
\country{Canada}
}

\author{Toghrul Karimov}
\affiliation{%
  \institution{MPI-SWS, Saarland Informatics Campus}
 \country{Germany}
 }

\author{Dominique Perrin}
\affiliation{%
  \institution{Universit\'e Gustave Eiffel}
\country{France}
}

\author{Mihir Vahanwala}
\affiliation{%
  \institution{MPI-SWS, Saarland Informatics Campus}
\country{Germany}
}

\AtEndPreamble{
    \theoremstyle{acmdefinition}
    \newtheorem{remark}[theorem]{Remark}
}

\begin{document}

\begin{abstract}
Suppose we have a deterministic finite-state transducer $\mathcal{A}$ and an infinite word $x$, and run $\mathcal{A}$ on $x$ to obtain an infinite word $\mathcal{A}(x)$.
Which properties of $x$ are guaranteed to also hold for $\mathcal{A}(x)$?
In this paper, we study this preservation question for various well-known combinatorial properties, e.g., recurrence, being morphic, and having factor frequencies.
The celebrated Krohn-Rhodes theorem provides the framework for proving our preservation results, and our techniques are based on the ergodic theory of symbolic dynamical systems, i.e., shift spaces.
\end{abstract}



\maketitle

\section{Introduction}

Words with elegant combinatorial properties often accurately capture the behaviour of symbolic dynamical systems and natural phenomena. Classes of self-similar words like \emph{morphic words}  \cite{Queffelec2010,Fog02}, and in particular, \emph{automatic sequences} \cite{Allouche_Shallit_2003} have profound connections with logic and automata theory \cite{Shallit:23}. Such words have been of interest for over a century: e.g., the  ubiquitous Thue-Morse word (as qualified  in   \cite{All99}) was defined to study a problem concerning geodesics  on surfaces \cite[Chap.~5]{Fog02}. It is moreover  an example of a \emph{uniformly recurrent} word, i.e., every factor occurs infinitely often (this first condition makes the word \emph{recurrent}) and there is a bound on the gaps between its consecutive occurrences (this additional condition makes the word \emph{uniformly} recurrent). This \emph{almost periodic} nature of uniformly recurrent words has been used to model quasicrystals in the setting of aperiodic order  \cite{BG13}: one may consult \cite{Fog02} and its bibliography for such applications in physics.

A perennially relevant class of uniformly recurrent words is that of \emph{toric words} \cite{berthe2025monadic}.
They have traditionally attracted interest in the study of dynamical systems and ergodic theory (see \cite[Chap.~1]{walters2000introduction} and \cite{weyl1916gleichverteilung}). More recently, toric words have been identified to accurately capture the behaviour of \emph{linear loops} in program verification \cite{Karimov2023}. They also play a central role in the decidability of the monadic second-order (MSO) theory of the structure ${\langle \mathbb{N}; <, a_1^\mathbb{N}, \ldots, a_k^\mathbb{N}\rangle}$ as established by \cite{berthe2024decidability}.

Intuitively, if the infinite  word $x$ is a trace of a system (e.g., the coding of an orbit of a point under the action of a dynamical system), and $\mathcal{A}$ is a deterministic finite-state transducer, then $\mathcal{A}(x)$ is obtained by augmenting $x$ with annotations with respect to some specification provided by the transducer. It is natural to ask: are combinatorial properties of $x$ preserved even after adding such annotations, and to what extent? More precisely, this paper considers the following concrete questions.
\begin{enumerate}
    \item If $x$ is recurrent, then is $\mathcal{A}(x)$ also recurrent? Can the \emph{recurrence function} of $\mathcal{A}(x)$ be described and computed?
    \item If $x$ admits \emph{factor frequencies}, then does $\mathcal{A}(x)$? If these frequencies are effective for $x$, are they effective for $\mathcal{A}(x)$? 
\end{enumerate}

In the above, (i) the recurrence function $R_x\colon \mathbb{N} \rightarrow \mathbb{N} \cup \{\infty\}$ of a recurrent infinite word $x$ returns for each $n$, the maximum gap between two consecutive occurrences of a factor $u$ of length $n$; (ii) an infinite word $x$ admits factor frequencies if for every finite word~$u$, we have that $\lim_{N\rightarrow \infty} \frac{1}{N}\cdot |\{i \mid 0 \le i < N, x(i, i+|u|) = u\}|$ exists.  

A classical result regarding the preservation of factor frequencies by transducers is Agafonov's theorem (see e.g., \cite{agafonov-statement}), which states that if $\mathcal{A}$ is a deterministic \emph{oblivious subsequence-selecting} transducer, then $x$ is \emph{normal} if and only if $\mathcal{A}(x)$ is normal. A normal word $x \in \Sigma^\omega$ admits factor frequencies such that for any $u \in \Sigma^+$, the frequency of $u$ is $1/|\Sigma|^{|u|}$. Normality is a prerequisite for \emph{randomness}, and its preservation by transducers has been studied in \cite{CartonOrduna,Carton:2022}. This paper complements the above results by showing that transducers can also preserve the existence of factor frequencies by virtue of preserving combinatorial structure. 

We identify that the Krohn-Rhodes theorem \cite{meyer1969remarks} provides a convenient framework to answer our preservation-related questions. This is because it guarantees that $\mathcal{A}(x)$ can be obtained as a cascade $\sigma \circ \mathcal{B}_k \circ \cdots \circ \mathcal{B}_1(x)$, where $\mathcal{B}_1, \ldots, \mathcal{B}_k$ are ``simple'' letter-to-letter transducers whose underlying automata are either \emph{permutation automata} or \emph{reset automata}, and $\sigma$ is a letter-to-word substitution. In other words, it suffices to answer the above questions for substitutions (done in Lem.~\ref{lem:automata-suffice}), permutation automata, and reset automata. 

To answer the first question regarding recurrence, we use the Krohn-Rhodes framework to revisit the work of Pritykin \cite{pritykin}, who used a result of Sem\"enov et al. \cite{semenov-english} to show that if $x$ has a uniformly recurrent suffix, then so does $\mathcal{A}(x)$. The advantages of using the Krohn-Rhodes theorem are as follows. (i) We can prove a slightly stronger result, i.e., if $x$ has a recurrent suffix, then so does $\mathcal{A}(x)$. (ii) We can already give an effective answer in the affirmative to the second question regarding frequencies in the special case of uniformly recurrent words and \emph{counter-free} transducers. (iii) In the original case of uniform recurrence and arbitrary transducers, we can explicitly describe and compute the starting index of the uniformly recurrent suffix of $\mathcal{A}(x)$, and also the recurrence function of this suffix. In particular, if $x$ is \emph{linearly recurrent} (i.e., the recurrence function of $x$ grows linearly), then the suffix of $\mathcal{A}(x)$ is also linearly recurrent; likewise if $x$ is polynomially recurrent, then so is the suffix of $\mathcal{A}(x)$. Moreover, if $x$ is (linearly, polynomially, uniformly) recurrent and the underlying automaton of $\mathcal{A}$ is a permutation automaton, then $\mathcal{A}(x)$ itself possesses the corresponding recurrence property. These preservation results are stated and proved in Sec.~\ref{sec::preservation-of-recurrence}, and will be useful to treat the second question regarding factor frequencies in greater generality. 

What about preservation of having factor frequencies?
The astute reader might already provide the counterexample of the (non-recurrent automatic) word $x \in \{0, 1\}^\omega$ such that $x(n) = 1$ if and only if $n$ is a power of $3$, and the two-state transducer $\mathcal{A}$ that permutes the states upon reading the letter $1$ and prints the current state. The word $x$ admits factor frequencies but $\mathcal{A}(x)$ does not. Applying Lem.~\ref{lem:automatic}, we deduce that $\mathcal{A}(x)$ is also automatic, and thus automatic (and hence morphic) words do not necessarily admit factor frequencies.

The question of what structure a word needs in order to admit factor frequencies is indeed profound. We note in Lem.~\ref{lem:generic} the intrinsic connection to symbolic dynamics and ergodic theory: a word $x$ admits factor frequencies if and only if it is \emph{generic} for an \emph{invariant measure} on the \emph{shift space} $X$ that it generates. Normal words, for instance, are by definition generic for the uniform Borel probability measure on the full shift $\Sigma^\omega$. We therefore seek to understand when we can guarantee that $\mathcal{A}(x)$ is also a generic point for some invariant measure on the shift space $Y$ it generates.

The most reliable way of doing so is to ensure that the shift $Y$ generated by $\mathcal{A}(x)$ is \emph{uniquely ergodic}, i.e., it admits a single invariant measure; by Oxtoby's theorem, all words in $Y$ would then be generic for this measure. It does not suffice to merely require that $x$ generate a uniquely ergodic shift: this condition is met by the counterexample above. Nevertheless, if a shift space is uniquely ergodic, then the support of its invariant measure will be a \emph{minimal} shift space, i.e., a shift space generated by a uniformly recurrent word. In the above counterexample, this support is the singleton shift space generated by $0^\omega$. We further prove that all words in a minimal uniquely ergodic shift space admit computable factor frequencies (Lem.~\ref{lem:ur-ue-compfreq}).

It is tempting to try to establish that if $x$ generates a minimal shift space $X$ that is uniquely ergodic, then $\mathcal{A}(x)$ generates a shift space $Y$ that is uniquely ergodic. Unfortunately, as the symbolic traces of ingenious counterexamples shown by Veech \cite{veech1969strict}, Sataev \cite{sataev1975number}, Chaika \cite{chaika-counter}, and Guenais \& Parreau \cite{guenais605250valeurs} demonstrate, this is not the case. Intuitively, the introduction of ``annotations'' gives rise to words whose factor frequencies ``oscillate'' between different limits, or, in technical terms, a generic point in $X$ need not lift to a generic point in $Y$.

What additional structure must we impose on $x$ in order to ensure that $\mathcal{A}(x)$ generates a uniquely ergodic shift space? The following is an ideal starting point. A \emph{primitive morphic} word $x$ is a uniformly recurrent word that is morphic, i.e., obtained by applying a substitution $\tau$ to the fixed point of a non-trivial substitution $\sigma$. Primitive morphic words are in fact linearly recurrent, and are classic examples of words that generate uniquely ergodic shifts. We show that if $x$ has a primitive morphic suffix, then so does $\mathcal{A}(x)$, and this result is fully effective, i.e., we can compute the factor frequencies of $\mathcal{A}(x)$; see~Thm.~\ref{thm:morphic-preservation}. We mention that we incidentally prove that more general notions of self-similarity are preserved while obtaining $\mathcal{A}(x)$ from $x$ (see Lem.~\ref{lem:self-similar}).  

Our main contribution lies in proving that, in order to guarantee preservation of having factor frequencies, it suffices to impose far less structure on $x$ than that implied by being primitive morphic. More precisely, we consider an ergodic-theoretic property known as \emph{Boshernitzan's condition}, or \emph{Condition (B)} (Def.~\ref{def:Boshernitzan}) on \emph{minimal} shifts. If a minimal shift $X$ satisfies Boshernitzan's condition, then it is uniquely ergodic. Boshernitzan's condition is rather relaxed: it is known that all linearly recurrent words, as well as all \emph{Sturmian} words, generate shifts that satisfy it. We show that if $x$ is uniformly recurrent and generates a shift that satisfies Boshernitzan's condition, then the uniformly recurrent suffix of $\mathcal{A}(x)$ also generates a shift that satisfies Boshernitzan's condition, and $\mathcal{A}(x)$ thus admits factor frequencies. This main result is Thm.~\ref{thm:boshernitzan-preserve} in the paper. As a special case, we show that the factor frequencies of $\mathcal{A}(x)$ are computable when $x$ is a computable  Sturmian word and the underlying automaton of $\mathcal{A}$ is a permutation automaton (Cor.~\ref{cor:Sturmian}).

Our preservation theorems are summarised in Fig.~\ref{fig:summary}, which depicts the relations between combinatorial properties, whether $\mathcal{A}(x)$ inherits them from $x$, and whether the corresponding preservation theorem is effective.

\begin{figure}
    \centering
    \includegraphics[width=\linewidth]{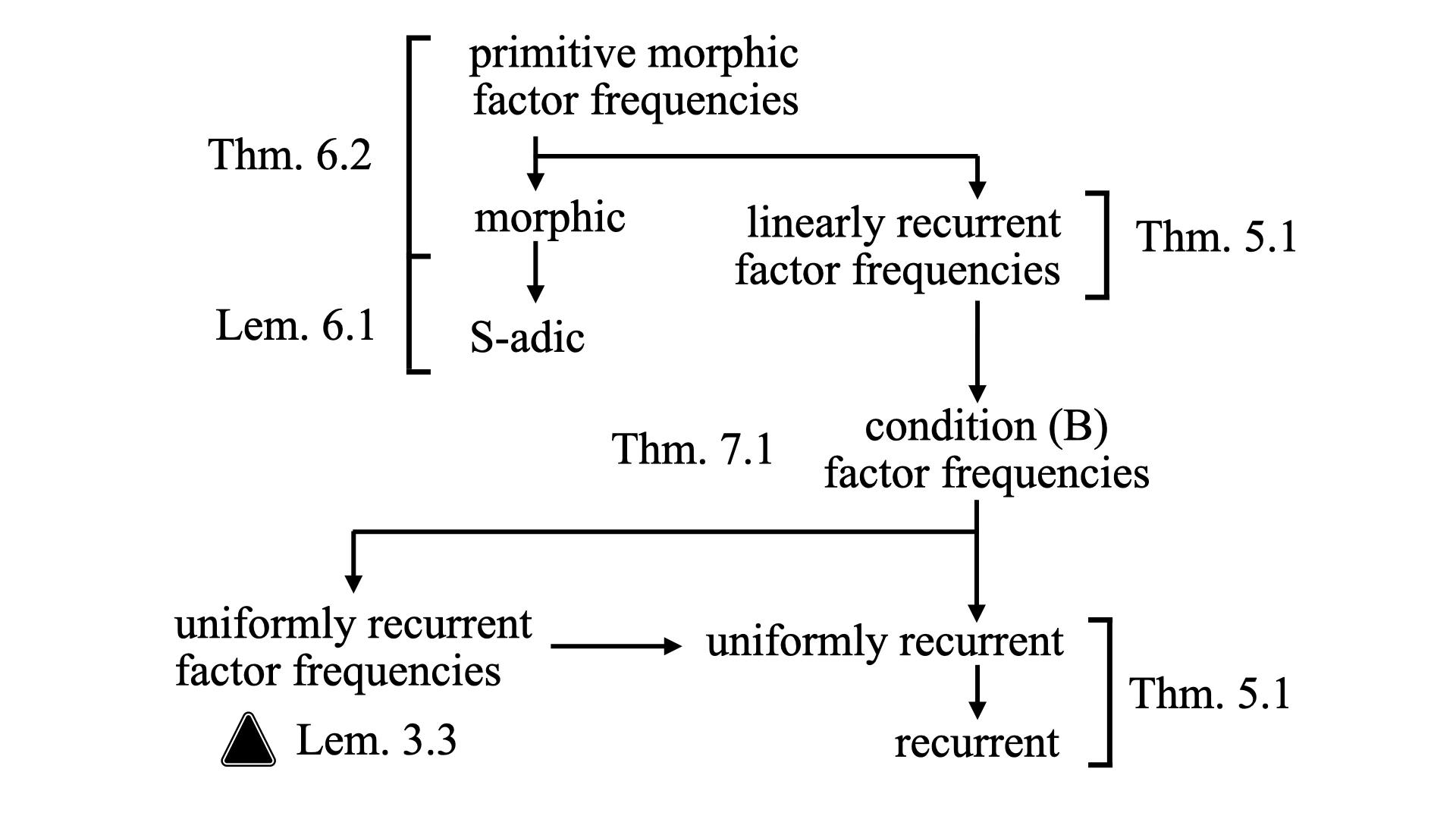}
    \caption{Summary of main results. Arrow from block $P$ to block $Q$ means that the properties in $P$ imply those in $Q$. We have effective preservation results for all properties in the diagram, with one exception. Uniformly recurrent $x$ admitting factor frequencies does not guarantee that $\mathcal{A}(x)$ does.}
    \label{fig:summary}
\end{figure}

\section{Preliminaries}
\subsection{Notation and Terminology}
Throughout this paper, we shall work with words over finite non-empty alphabets, which will usually be denoted by $\Sigma$ and sometimes by $\Gamma$. The set of infinite words over $\Sigma$ (indexed by $\mathbb{N}$) is denoted by $\Sigma^\omega$, the set of finite words is denoted by $\Sigma^*$, the empty word is denoted by $\varepsilon$, and the set of finite non-empty words is denoted by $\Sigma^+$. For a word $x$, the notation $x(i)$ denotes the letter in the $i$-th position of $x$, the notation $x(i, j)$ denotes the finite subword, or \emph{factor} of $x$ starting in position $i$ and ending at position $j-1$ (i.e., the length of the factor is $j-i$). When $u = x(i, j)$, we say that $u$ has an \emph{occurrence} in $x$ at index $i$. If $x$ is an infinite word, the notation $x(i, \infty)$ denotes the suffix of $x$ starting at position $i$. We shall denote the length of a finite word $u$ as $|u|$. The factor complexity of a word $x$ is the function $P_x$ that counts for each $n$, the number of distinct length-$n$ factors of $x$. We denote by $\mathcal{L}(x)$ the \emph{factor language} of $x$, i.e., the set of finite factors $u$ of $x$. We shall sometimes use $\mathcal{L}_n(x)$ as a convenient shorthand for $\mathcal{L}(x) \cap \Sigma^n$.

A word $x \in \Sigma^\omega$ is said to be \emph{recurrent} if every factor $u$ of $x$ occurs infinitely often. In this case, we can define the set of \emph{return words} to a factor $u$ of $x$: these are words $r$ such that $ru$ is a factor of $x$, and $ru$ has exactly two occurrences of $u$, once as a prefix, and once as a suffix. A recurrent word $x$ is said to be \emph{uniformly recurrent} if for every factor $u$, its set of return words $\mathcal{R}_x(u)$ is finite. For a uniformly recurrent word $x$, we can define a \emph{return-time} (or \emph{recurrence}) function $R_x\colon \mathbb{N}\rightarrow \mathbb{N}$, such that $R_x(n)$ gives the length of the longest return word to a length-$n$ factor of $x$. A uniformly recurrent word $x$ is said to be \emph{linearly recurrent} if $R_x(n) \in O(n)$. A helpful perspective is as follows: in a uniformly recurrent word $x$, for any factor $u$, the gaps between its consecutive occurrences are bounded. If $x$ is linearly recurrent, the bound is linear in $|u|$. If for a uniformly (respectively, linearly) recurrent word $x \in \Sigma^\omega$, given any $n$, we can compute $x(n)$ and $R_x(n)$, we say that $x$ is \emph{effectively} uniformly (respectively, linearly) recurrent.

\begin{remark}
    \label{remark:compute-fac-lang}
    If $x(n)$ and $R_x(n)$ are computable, we can compute $\mathcal{L}_n(x)$ as the set of length-$n$ factors of $x(0, n + R_x(n))$. This is because if the occurrence of $u$ at index $i$ is its first occurrence, then $v = x(0, i)$ must be a suffix of some return word $r$ to $u$, since indeed, $vu$ is a recurrent factor of $x$. In particular, $R_x$ gives a bound on the index of the first occurrence of a factor. Conversely, if $x$ is uniformly recurrent and $x(n)$ and $\mathcal{L}_n(x)$ are computable given any $n$, then we can compute $R_x(n)$ by enumerating $\mathcal{L}(x)$ till we find an $N$ such that all length-$N$ factors contain at least two occurrences of each length-$n$ factor.  
\end{remark}

For arbitrary recurrent words, the set of return words to a factor is not necessarily finite. In this general setting, we say that a word $x$ is \emph{effectively recurrent} if it is recurrent, and for any regular language $L$, the following two problems are decidable: (A) Does there exist $N > 0$ such that $x(0, N) \in L$; (B) Does there exist non-empty $u \in L \cap \mathcal{L}(x)$? Due to a result of Sem\"enov (see e.g., \cite[Thm.~5]{colcombet-semenov}), this definition is equivalent to asserting that $x$ is recurrent and has a decidable monadic-second order (MSO) theory.\footnote{Technically, just the decidability of problem (B) suffices; a justification can be elicited from \cite[Cor.~3]{colcombet-semenov}. There are, however, examples of words that are computable and recurrent, have recursive factor language, but their attendant MSO theory is undecidable. One such word $x$ is constructed as follows. Map a pair $(m, k) \in \mathbb{N} \times \mathbb{N}$ to $a^{m+1} b^{k+1} c$ if the $m$-th Turing machine does not halt within $k$ steps, and to $a^m b^k h$ otherwise. Fix an enumeration $v_0, v_1, \ldots$ of such words, and define $x = \lim_{n \rightarrow \infty} u_n$, where $u_0 = v_0$, and $u_{n+1} = u_n v_{n+1} u_n$. The MSO theory of $x$ subsumes the Halting Problem.} Recall that it is elementary to show that if $x$ has a decidable MSO theory, then so does $\mathcal{A}(x)$ for any transducer $\mathcal{A}$ (see, e.g., \cite[Lem.~4.5]{berthe2025monadic}).

Let $u$ be a factor of an infinite word $x$. The \emph{frequency} of $u$ in $x$, denoted $f_x(u)$, is given by $\lim_{N \rightarrow \infty} \frac{1}{N}{|\{i \mid x(i, i+|u|) = u, ~i < N\}|}$ (i.e., the limiting fraction of indices which mark an occurrence of $u$), and is defined if the limit exists. We say that a word $x \in \Sigma^\omega$ \emph{admits factor frequencies} if for every recurrent factor $u$ of $x$, the above limit exists. If furthermore, given any $u$, this limit can be computed, then we say that $x$ admits 
\emph{computable} factor frequencies. Formally, there is a Turing machine which, given factor $u$ and tolerance $\delta$ as input, computes a rational number $\rho$ such that $|f_x(u) - \rho| < \delta$. Moreover,  for any given $u$,  if  the convergence of the sequence $\frac{1}{N}{|\{i \mid x(i, i+|u|) = u, ~ k \leq i < N+k\}|}$  toward $ f_x(u)$
holds uniformly in $k$, then $x$ is said to admit \emph{uniform frequencies} (see e.g. \cite{Fog02,Queffelec2010}).

\subsection{Shift Spaces and Measures}
\label{sec::shift-measures-ergodicity}
We shall use ergodic theory to study which subsets of uniformly recurrent words that admit factor frequencies are closed under transduction. Shift spaces will serve as our underlying dynamical systems in order to do so. We also remark that combinatorial properties (e.g., recurrence) are often studied for shift spaces, but it is usually straightforward to obtain analogous results for words. 

The shift operator $T$ maps a word $x \in \Sigma^\omega$ to $x(1, \infty)$, i.e., it deletes the starting letter. 
We write $T^n x$ for the $n$-fold application of $T$ to $x \in \Sigma^\omega$.
We endow $\Sigma^\omega$ with the usual product topology, thus making it a compact metric space. 
A \emph{shift system} (also called a \emph{shift}) is the dynamical system $(X,T)$ where $X \subseteq \Sigma^\omega$ is closed and satisfies $TX \subseteq X$. 
When the dynamics is clear from the context, we also refer to $X$ as a shift.

Given a word $x$, the \emph{shift generated by} $x$ is the the topological closure of the orbit $\{T^n x \mid n \in \mathbb{N}\}$ of $x$ under $T$. 
For a shift $X$, we define $\mathcal{L}(X) = \{u \mid u \in \mathcal{L}(x'), ~x'\in X\}$, and as before, we declare $\mathcal{L}_n(X)$ to be shorthand for $\mathcal{L}(X) \cap \Sigma^n$. 
We have that if $X$ is the shift generated by $x$, then $\mathcal{L}(X) = \mathcal{L}(x)$, and $X = \{x' \mid \mathcal{L}(x') \subseteq \mathcal{L}(x)\}$ \cite[Prop.~4.6]{Queffelec2010}. 
A shift $X$ is said to be \emph{minimal} if it does not have a non-empty proper closed subset $Y$ such that $TY \subseteq Y$.
Therefore, the shift generated by any $x'$ in a minimal shift $X$ must be $X$ itself. 
We have that a word $x$ generates a minimal shift if and only if it is uniformly recurrent \cite[Prop.~4.7]{Queffelec2010}.

A cylinder of $X \subseteq \Sigma^\omega$ is a set of the form $\{ux' \colon x' \in \Sigma^\omega \} \cap X$ where $u \in \Sigma^*$.
The cylinder generated by $u$ will be denoted by $[u]_X$; we omit the superscript $X$ when it is clear from the context.
The cylinders are clopen (i.e., both closed and open, and in particular are the open balls of $X$), and a subset of $X$ is clopen if and only if it is a finite union of cylinders. Since there are countably many cylinders, any open set is a countable union of cylinders.

We refer the reader to \cite[App.~B.5]{DurandPerrin2022} or a standard text such as \cite{Bauer2001} for the basic concepts of measure spaces. 
Given a topological space $X$, its collection of \emph{Borel-measurable sets}, or simply \emph{Borel sets}, is the the $\sigma$-algebra generated by the open subsets of $X$. If $X$ is a shift space, the $\sigma$-algebra can equivalently be generated by the cylinders of $X$.

A (Borel probability) measure $\mu$ on $X$ maps each Borel set $B$ to $\mu(B) \in [0, 1]$, satisfies $\mu(X) = 1$, and is countably additive, i.e., for any countable collection $(B_n)_{n=0}^\infty$ of disjoint Borel sets we have $\mu(\bigcup_{n} B_n) = \sum_n \mu(B_n)$. 

A function $\pi \colon Y \rightarrow X$ is called \emph{Borel measurable} if for every Borel set $B \subseteq X$, the pre-image $\pi^{-1}(B)$ is a Borel set of~$Y$. 
In particular, if $\pi$ is continuous, then it is Borel measurable. Indeed, if $B$ is open, then by continuity $\pi^{-1}(B)$ is open; and we have that $\pi^{-1}(X \setminus B) = Y \setminus \pi^{-1}(B)$ and $\pi^{-1}(\bigcup_i B_i) = \bigcup_i\pi^{-1}(B_i)$. In other words, the family of sets $B$ for which $\pi^{-1}(B)$ is Borel contains the open sets of $X$ and forms a $\sigma$-algebra; in particular it must contain the $\sigma$-algebra generated by the family of open sets of $X$, i.e., the Borel sets of $X$.
A measurable function $\pi\colon Y \rightarrow X$ \emph{projects} a Borel measure $\nu$ on $Y$ to a Borel measure $\mu$ on $X$ as $\mu = \nu \circ \pi^{-1}$.

An \emph{invariant measure} $\mu$ of a compact dynamical system $(X,T)$ is a (Borel probability) measure $\mu$ such that $\mu(T^{-1} B )= \mu(B)$ for every measurable $B \subseteq X$.
By the Krylov-Bogolyubov theorem \cite[3.8.4]{DurandPerrin2022}, every $(X,T)$ has an invariant measure.

In the case of shift spaces, we can obtain an invariant measure $\mu$ by defining a probability pre-measure $\mu_0$ on the Boolean algebra of cylinders, and then invoking the Carath\'eodory extension theorem \cite[Thm.~B.5.1]{DurandPerrin2022} to uniquely extend it to a Borel measure $\mu$, since the $\sigma$-algebras generated by cylinders and open sets of a shift space coincide. We require that the pre-measure $\mu_0$ evaluate to $1$ on $X = [\varepsilon]_X$, and satisfy the compatibility conditions
\begin{equation}
\mu_0([u]_X) = \sum_a\mu_0([ua]_X) = \sum_a \mu_0([au]_X).
    \label{eq:compatibility}
\end{equation}
Observe in particular that $\mu_0 \circ T^{-1}([u]_X) = \mu_0([u]_X)$, and more generally, for any set $U$ in the Boolean algebra of cylinders, $\mu_0(U) = \mu_0(T^{-1}U)$. The Borel measure $\mu$ is defined as $$\mu(B) = \inf_{(U_n)_n} \sum_n \mu_0(U_n),$$ where $(U_n)_n$ ranges over sequences of (disjoint) sets in the Boolean algebra of cylinders such that $B \subseteq \bigcup_n U_n$. 

Since $T$ is continuous, $\mu \circ T^{-1}$ will also be a Borel measure $\nu$. The corresponding pre-measure $\nu_0$ obtained by restricting $\nu$ to the Boolean algebra of cylinders, however, is identical to $\mu_0$. Thus by the uniqueness of the extension, $\nu = \mu \circ T^{-1} = \mu$, i.e., $\mu$ is indeed an invariant measure. 

We say that $x \in X$ is a \emph{generic point} of an invariant measure $\mu$ (or simply, $x$ is generic for $\mu$) if for every continuous $h \in C(X, \mathbb{R})$ (where $C(X, \mathbb{R})$ denotes the set of all continuous functions from $X$ to $\mathbb{R}$),
\[
    \lim_{N\to\infty} \frac 1 N \sum_{i=0}^{N-1} h(T^i x) = \int h d\mu.
\]

We refer the reader to \cite[App.~B.5.2]{DurandPerrin2022} for the precise formal definition of the integral.

\subsection{Morphic Words and Primitivity}
A \emph{substitution} is a map $\sigma\colon \Sigma \rightarrow \Gamma^*$, and is extended to $\Sigma^*$ and $\Sigma^\omega$ in the obvious way, i.e., concatenation. A substitution $\sigma\colon \Sigma \rightarrow \Sigma^*$ is said to be \emph{primitive} if there exists a positive integer $m$ such that for every letter $a \in \Sigma$, all letters occur in $\sigma^m(a)$. A substitution $\sigma$ is said to be \emph{non-erasing} if there does not exist a letter $a$ such that $\sigma(a) = \varepsilon$.

A word $x \in \Sigma^\omega$ (we assume $x$ contains all the letters in $\Sigma$) is called substitutive if there exists a substitution $\sigma$ which is non-trivial (i.e, not the identity) over the letters of $x$, such that $\sigma(x) = x$. If $\sigma$ is primitive, then $x$ is said to be a primitive substitutive word. As an example, the Fibonacci word $x = 0100101001001\cdots$ is primitive substitutive, and is the fixed point of the substitution that maps $0$ to $01$ and $1$ to $0$.

A word $y \in \Gamma^\omega$ is called \emph{morphic} if there exists a substitutive word $x \in \Sigma^\omega$ and a substitution $\tau\colon \Sigma \rightarrow \Gamma^*$ such that $y = \tau(x)$. If $x$ is primitive substitutive, then $y$ is said to be primitive morphic. Morphic words are presented as $\sigma, \tau$, a prefix $u$ such that $\sigma(u) = u$, and a letter $a$ such that $\sigma(a)$ prolongs $a$ (i.e., begins with $a$ and is more than one letter long). We can convert any representation into one such that the defining substitutions are non-erasing by restricting the alphabet of $\sigma$.

\begin{remark}
\label{remark:cobham}
    Some authors require $\tau$ to be a coding, i.e., $\tau\colon \Sigma \rightarrow \Gamma$ in the above definition. A result of Cobham observes that this is not a restriction. Indeed, these definitions can (effectively) be used interchangeably, including in the case of primitive morphic words. 
    The idea to prove the equivalence (presented in \cite[Prop.~17]{durand-hd0l}) is to define a substitution $\hat\sigma$ over an auxiliary alphabet $\hat\Sigma$, constructed by taking $|\tau(a)|$ copies of each letter $a \in \Sigma$.
\end{remark}

There are several other sources the reader can consult for a detailed technical exposition of this remark, see e.g., \cite[Cor.~7.7.5]{Allouche_Shallit_2003}, \cite[Thm.~3.8]{durand-hd0l-decide} (Thm.~9 in the arXiv version), or \cite{honkala}. For completeness, we sketch the idea of \cite[Prop.~17]{durand-hd0l} here. For $0 \le i< |\tau(a)| - 1$, the substitution $\hat\sigma$ maps the $i$-th copy $a^{(i)}$ of $a$ to the concatenation $b_i^{(0)}b_i^{(1)} \cdots$ of all the copies of the $i$-th letter $b_i$ of $\sigma(a)$ (if $|\sigma(a)| \le i$, then the image is the empty word). The last copy of $a$ is mapped similarly to the copies of the remaining letters of $\sigma(a)$. We then get $\hat x $, a ``stuttering'' version of $x$, as the fixed point of $\hat \sigma$ iterated on $c^{(0)}$, the foremost copy of the starting letter of $x$. Finally, $y$ is obtained as $\hat\tau(\hat x)$, where $\hat \tau$ maps $a^{(i)}$ to the $i$-th letter of $\tau(a)$.

\begin{remark}
\label{remark:recurrentsuffix}
  It is straightforward to give an effective proof of the fact that any suffix of a morphic word is morphic \cite[Thm.~7.6.1]{Allouche_Shallit_2003}. It is also known how to decide whether a morphic word (presented as the image under $\tau$ of the fixed point of $\sigma$) is uniformly recurrent; should the decision be yes, the procedure computes letter-to-letter $\tau'$ and primitive $\sigma'$ such that the input is the image under $\tau'$ of the fixed point of $\sigma'$ \cite[Thm.~1, Thm.~3, Sec.~4, Sec.~5]{durandmorphicdecidability} (see also \cite[Lem.~4]{durand-hd0l} and the statement of \cite[Thm.~24(1)]{stability-return}). 
  It is well known that primitive substitutive words are, in fact, \emph{linearly} recurrent \cite[Prop.~25]{dhs1999}. By Lem.~\ref{lem:morphic-freq} below, the linear recurrence function and factor frequencies are effective. These properties extend to primitive morphic words by Lem.~\ref{lem:automata-suffice}. The proof of effectiveness is implicit is classical texbook, we provide it for completeness.
\end{remark}

\begin{lemma}
\label{lem:morphic-freq}
    Let $x \in \Sigma^\omega$ be a substitutive word obtained as the fixed point of a primitive substitution $\sigma$, and let $y = \tau(x)$ be a primitive morphic word. We have that $x, y$ are effectively linearly recurrent and admit computable factor frequencies.  
\end{lemma}

\begin{proof}
    We prove the lemma for the primitive substitutive $x$, and defer the proof for the primitive morphic $y$ to Lem.~\ref{lem:automata-suffice}, whence it follows immediately. In the proof, we shall associate to the substitution $\sigma$ an adjacency matrix $M$, whose $(j, i)$-th entry records the number of times the $j$-th letter $a_j$ appears in the image $\sigma(a_i)$ of the $i$-th letter $a_i$. Because $\sigma$ is primitive, we can assume up to replacing $\sigma$ with some $\sigma^k$ that $M$ has strictly positive entries, which allows us to use Perron--Frobenius theory.

    We first show that $x$ is effectively linearly recurrent, i.e., $R_x(n) \le Ln$ for some computable $L$. The proof of linear recurrence in \cite[Prop.~25]{dhs1999} shows that $L = C \cdot \max_{a\in \Sigma}|\sigma(a)| \cdot R_x(2)$, where $C$ is such that for all $n \ge 1$, $\max_{a \in \Sigma} |\sigma^n(a)| \le C \cdot \min_{a \in \Sigma}|\sigma^n(a)|$. 
    
    It suffices to take $C = (\max_{i,j}(v_i/v_j))^2$, where $(v_i)_{i=1}^{|\Sigma|}$ are the strictly positive entries of a left eigenvector $v^\top$ of $M$ corresponding to the Perron-Frobenius eigenvalue $\rho$ (i.e., the positive one with maximum absolute value). Indeed, for the $i$-th letter $a$, $|\sigma^n(a)| = \sum_j (M^n)_{j, i}$. The expression on the right is the sum of entries in the $i$-th column of $M^n$, and is lower bounded by $\sum_j (v_j/\max_j v_j) \cdot (M^n)_{j, i}$, and upper bounded by $\sum_j (v_j/\min v_j) \cdot (M^n)_{j, i}$. These bounds respectively evaluate to the $i$-th entries of $(v^\top M^n)/(\max_j v_j)$ and $(v^\top M^n)/(\min_j v_j)$. Since $v^\top$ is an eigenvector, we get that $$\rho^j (\min_j v_j)/(\max_j v_j) \le |\sigma^n(a)| \le \rho^j (\max_j v_j)/(\min_j v_j).$$

    We now show how to compute (an upper bound on) $R_x(2)$. First, we can compute $\mathcal{L}_x(2)$ by starting with $x(0, 2)$, and recording the length-$2$ factors that are produced upon repeatedly applying $\sigma$. The set of factors will saturate within $|\Sigma|^2$ iterations. In particular, we can find $k$ large enough such that for all $a \in \Sigma$, $\sigma^k(a)$ contains every factor in $\mathcal{L}_x(2)$. Clearly, $R_x(2) \le 2 \max_a |\sigma^k(a)|$. This proves effective linear recurrence.
    
    We note that the textbook discussions in \cite[Chap.~5.4]{Queffelec2010} and \cite[Chap.~3.8.5]{DurandPerrin2022} prove $x$ has computable factor frequencies. We present the techniques here for clarity in exposition.
    Recall the adjacency matrix $M$ of $\sigma$. It follows that the frequency of the $i$-th letter $a$ in $x$ is the $i$-th entry of the eigenvector corresponding to the Frobenius eigenvalue.

    This same idea is extrapolated to describe factor frequencies as follows. We define an alphabet $\Sigma_n$, whose letters correspond to words in $\mathcal{L}_n(x)$ (see Rmk.~\ref{remark:compute-fac-lang} for how to compute this alphabet), and the word $x_n \in \Sigma_n^\omega$, where $x_n(i)$ corresponds to the factor $x(i, i+n)$. We define the substitution $\sigma_n$ such that $\sigma_n(u)$ is the ordered list of the first $|\sigma(u(0))|$ length-$n$ factors of $\sigma(u)$. Crucially, this substitution on $\Sigma_n$ is primitive \cite[Lem.~5.3]{Queffelec2010}, and we can again use Perron-Frobenius theory to obtain the frequency of factor $u \in \Sigma^n$, which is a letter in $\Sigma_n$.
\end{proof}

Morphic words, by their definition, can be regarded as being self-similar. This notion of self-similarity can be generalised as follows. 

\begin{definition}
\label{def:S-adic}
Let $S$ be a set of substitutions. A word $x \in \Sigma^\omega$ is said to be $S$-adic if there exists a sequence $\boldsymbol{\sigma} = (\sigma_n)_{n=0}^\infty$ of substitutions from $S$, and words $x^{(0)} = x, x^{(1)}, x^{(2)}, \ldots$ such that for all $n$, $x^{(n)} = \sigma_n\left( x^{(n+1)}\right)$. The sequence $\boldsymbol{\sigma}$ is called the directive sequence of the word $x$, and we say that $x$ is directed by $\boldsymbol{\sigma}$. 
\end{definition}

Note that substitutive and morphic words are special cases of $S$-adic words where the directive sequence is periodic and eventually periodic, respectively. 
The most common example of a class of $S$-adic words is the class of Sturmian words (see e.g. \cite{Fog02,berthe-survey}): these are words over the binary alphabet with a factor complexity of $n+1$. 
Let $S = \{\lambda_0, \lambda_1, \rho_0, \rho_1\}$, where $\lambda_i(i)= i, \lambda_i(j) = ij, \rho_i(i) = i, \rho_i(j)= ji$. It is well known that any Sturmian word is $S$-adic with $S$ as above, and an analogous statement can also be made for the generalisation to Arnoux-Rauzy words \cite{Glen2009episturm}. As an example, the Fibonacci word is a Sturmian word, and the sequence of substitutions is $(\lambda_0 \lambda_1)^\omega$.

\section{Frequencies and Ergodic Theory}
In this section, we shall prove the following, and use it to motivate the conditions we impose on a word $x$ (i.e., Boshernitzan's condition) in order to ensure that $\mathcal{A}(x)$ admits factor frequencies.

\begin{lemma}
\label{lem:generic}
    A word $x$ admits factor frequencies if and only if $x$ is generic for an invariant measure $\mu$ on the shift $X$ that it generates. Should this be the case, we have $f_x(u) = \mu([u]_X)$. 
\end{lemma}
\begin{proof}
The ``if'' implication is obvious: we simply choose the continuous function $h$ to be the indicator $\mathbf{1}_u$ which evaluates to $1$ on the cylinder $[u]_X$ and $0$ elsewhere. 

Conversely, suppose $x$ admits factor frequencies and generates the shift $X$. Recall that $\mathcal{L}(X) = \mathcal{L}(x)$, and hence the factor frequencies of $x$ defines a pre-measure $\mu_0$ on each cylinder in a way that satisfies the compatibility condition (\ref{eq:compatibility}). The pre-measure $\mu_0$ is thus invariant, and extends to an invariant measure $\mu$. We now need to prove that $x$ is generic for $\mu$, i.e., for every continuous $h$, we have $\int h d\mu = \lim_{N\rightarrow \infty}\frac{1}{N}\sum_{i=0}^{N-1}h(T^ix)$.

We have a special case where the Stone-Weierstraß theorem applies. Observe that since $X$ is compact, the function $h$ is uniformly continuous, admits a modulus of continuity, and we can define the following sequence $(g_M)_M$ of functions that converges uniformly to $h$ (i.e., for every $\delta$ there exists $M_\delta$ such that for all $M \ge M_\delta$ we have $\max_X |g_M - h| \le \max_{u \in \mathcal{L}_M(X)}(\max_{[u]_X} h - \min_{[u]_X} h ) < \delta$):

$$
g_M = \sum_{u \in \mathcal{L}_M(X)} \max_{[u]_X} h \cdot \mathbf{1}_{u}.
$$
By the dominated convergence theorem \cite[Thm.~B.5.3]{DurandPerrin2022}, we also have that $\int h d\mu = \lim_{M \rightarrow\infty} \int g_M d\mu$.

In particular, for every $\delta > 0$, we can choose $M$ large enough such that $$\left|\int h d\mu - \int g_Md\mu\right| < \delta/3,$$ and for all $N$, we have by uniform convergence that $$\left| \frac{1}{N}\sum_{i=0}^N g_M(T^i x) -  \frac{1}{N}\sum_{i=0}^N h(T^i x)\right| < \delta/3.$$

Since $x$ admits factor frequencies, it is also clear that for each $M$, we have $\lim_{N\rightarrow \infty}\frac{1}{N}\sum_{i=0}^{N-1}g_M(T^i x) = \int g_M d \mu$, and in particular for our chosen $M$, for every $\delta$ there exists $N_\delta$ such that for all $N \ge N_\delta$, we have $$\left|\int g_M d \mu - \frac{1}{N}\sum_{i=0}^{N-1}g_M(T^i x) \right| < \delta/3.$$

Adding our inequalities together and applying the triangle inequality, we get that for every $\delta$ there exists $N_\delta$ such that for all $N \ge N_\delta$, we have $\left|   \int h d\mu - \frac{1}{N}\sum_{i=0}^N h(T^i x)\right| < \delta$, or in other words, $\int h d\mu = \lim_{N\rightarrow \infty}\frac{1}{N}\sum_{i=0}^{N-1}h(T^ix)$, as desired.
\end{proof}

We now describe measures for which ``almost'' all points are generic. An invariant measure $\mu$ is called \emph{ergodic} when the following condition is met: if a Borel set $B$ is such that $T^{-1}B = B$, then $\mu(B)$ is either $0$ or $1$. Since $X$ is a compact metric space, $(X, T)$ is guaranteed to have an ergodic measure \cite[Prop.~3.8.9]{DurandPerrin2022}. The Birkhoff ergodic theorem \cite[Thm.~3.8.5]{DurandPerrin2022} states that if $\mu$ is an ergodic measure, then for every integrable function $h$, the set of points $x$ for which\footnote{We overload inequality to also include the case where the limit does not exist.} $\int h d \mu \ne \lim_{N\rightarrow \infty} \frac{1}{N}\sum_{i=0}^{N-1}h(T^i x)$ has $\mu$-measure $0$. Thus, for each $u \in \mathcal{L}(X)$, the set $B_u = \{x \mid f_x(u) \ne \mu([u]_X)\}$ has $\mu$-measure $0$. By Lem.~\ref{lem:generic}, the complement of the countable union of all such $B_u$ is precisely the set of points that are generic for $\mu$. Thus, the set $G$ of generic points for $\mu$ has $\mu$-measure $1$. If $(X,T)$ is minimal, we additionally have that $G$ is dense (in the topological sense) in $X$.

The exception of non-generic points does not arise in dynamical systems $(X, T)$ that are \emph{uniquely ergodic}, i.e., there exists only one invariant measure $\mu$, which is guaranteed to be ergodic \cite[Cor.~3.8.10]{DurandPerrin2022}. In this case, Oxtoby's theorem \cite[Thm.~4.3]{Queffelec2010} tells us that all points $x \in X$ are generic for $\mu$. Moreover, unique ergodicity is even equivalent to each word in the shift having \emph{uniform} factor frequencies \cite[Prop.~3.8.14]{DurandPerrin2022}. A word $x$ admits uniform factor frequencies if for every $v \in \mathcal{L}(x)$ and every $\delta$, there exists $M$ such that for every $m \ge M$ and every $u \in \mathcal{L}_m(x)$, we have $||u|_v/|u| - f_x(v)| < \delta$, where $|u|_v$ denotes the number of occurrences of $v$ in $u$. Note that for technical convenience, we may replace $|u|_v/|u|$ in the above by $|u|_v/(|u|-|v|+1)$, which we denote by $f_u(v)$. This alternate definition is also natural because the denominator is the number of indices of $u$ at which an occurrence of $v$ is possible. 

Note that if $(X, T)$ is uniquely ergodic, then the support $X'$ of the invariant measure $\mu$ is a minimal shift. Indeed, if it were to contain a shift $Y$, then $Y$ itself would admit an invariant measure $\nu$ \cite[Prop.~3.8.4]{DurandPerrin2022} which, when extended to $X'$, would be distinct from $\mu$: a contradiction. We henceforth focus on minimal shifts.

\begin{lemma}
\label{lem:ur-ue-compfreq}
    If an effectively uniformly recurrent word $x$ generates a uniquely ergodic shift $X$ with invariant measure $\mu$, then $x$ admits computable factor frequencies. 
\end{lemma}
\begin{proof}
    Given a factor $v$ of length $l$ and a tolerance $\delta$, we shall show that we can compute $L$ such that for all $u \in \mathcal{L}_{L+l-1}(x)$, we have $|f_x(v) - f_u(v)| < \delta$. It then remains to count the occurrences of $v$ in any appropriate $u$, e.g., $u = x(0, L+l-1)$. For notational convenience, in this proof we abbreviate $\mathcal{L}_{L+l-1}(x)$ as simply $\mathcal{L}$. 

    It suffices to show that for every $L \ge 1$, we have $\min_{u \in \mathcal{L}} f_u(v) \le f_x(v) \le \max_{u \in \mathcal{L}} f_u(v)$. Indeed, by the definition of uniform frequencies, the difference between the bounds is less than $\delta$ for some $L$ which can be found by brute enumeration (made possible by the fact that $x$ is effectively uniformly recurrent).

    We prove the above claim by showing that $f_x(v)$ is a convex combination of the values of $f_u(v)$; more precisely, we have $f_x(v) = \sum_{u \in \mathcal{L}} f_x(u) f_u(v)$. To that end, we define the continuous function $\chi$, which maps $x$ to $f_{x(0, L+l-1)}(v)$. We clearly have that $$\sum_{u \in \mathcal{L}} f_x(u) f_u(v) = \int \chi d\mu.$$ We shall conclude by showing that $\lim_{N\rightarrow \infty}\frac{1}{N}\sum_{i=0}^{N-1}\chi(T^i x) = f_x(v)$. 
    
    Observe that by definition, $\chi(x) = \frac{1}{L}\sum_{j=0}^{L-1} \mathbf{1}_v(T^j x)$. Thus, we can express $\sum_{i=0}^{N-1}\chi(T^i x)= \frac{1}{L}\sum_{i=0}^{N-1}\sum_{j=i}^{i+L-1}\mathbf{1}_v(T^j x)$. We can further rearrange the summation to obtain 
    $$
    \frac{1}{L}\left(L\sum_{i=L-1}^{N-L}\mathbf{1}_v(T^ix)+\sum_{i=0}^{L-2}(i+1)(\mathbf{1}_v(T^ix) + \mathbf{1}_v(T^{N-1-i}x))\right).
    $$
    Thus, $\lim_{N\rightarrow \infty}\frac{1}{N}\sum_{i=0}^{N-1}\chi(T^ix) = \lim_{N\rightarrow \infty}\frac{1}{N}\sum_{i=L-1}^{N-L}\mathbf{1}_v(T^i x)$. The latter sequence is $(N-2L)/N$ times the average of $N-2L$ terms, and indeed converges to the limit of the usual Birkhoff averages, i.e., to $f_x(v)$.
\end{proof}

The existence of factor frequencies is very susceptible to failure when words are drawn from minimal shift spaces that are not uniquely ergodic. Note that the set of all invariant measures of $(X,T)$ form a convex polytope whose extremal points are the ergodic measures \cite[Prop.~3.8.10]{DurandPerrin2022}.
In particular, if $(X,T)$ is not uniquely ergodic, it will have at least two distinct ergodic measures.

\begin{lemma}[Thm.~1.1 in \cite{dw2024}]
\label{lem:UE-equivalence}
    Let $X$ be a minimal shift space which is not uniquely ergodic. There exists a dense set $B \subseteq X$ of words that are not generic for any invariant measure, and hence do not admit factor frequencies.
\end{lemma}

\begin{remark}
    To summarise the discussion on the existence of factor frequencies, for a minimal shift $X$, we have the following. If $X$ is uniquely ergodic, then all $x \in X$ admit  (uniform and computable) factor frequencies, and moreover agree on the frequencies of each factor. On the other hand, if $X$ is not uniquely ergodic, then it will have a dense subset of words which do not admit factor frequencies, and a dense (by virtue of including the orbit of any generic point) subset of generic points for its ergodic measures. These generic points do admit factor frequencies, but they are not unanimous on the frequencies of all factors because they are generic for different measures.
\end{remark}

Let us now explain how to obtain the counterexamples to the preservation of unique ergodicity alluded to in the Introduction. 

\begin{lemma}
    There exists a uniformly recurrent word $x \in \{0, 1\}^\omega$ that admits factor frequencies, but $\mathcal{A}(x)$ does not, where $\mathcal{A}$ is the two-state transducer whose underlying automaton changes the state on reading $1$, stays in the same state on reading $0$, and prints $(q, b)$ on reading $b$, where $q$ is the current state.
\end{lemma}

\begin{proof}
    We shall choose $x$ to be the symbolic trace of a carefully constructed topological dynamical system. Chaika \cite{chaika-counter} constructs a minimal and uniquely ergodic dynamical system $(X', T')$ as follows: the space $X'$ is the unit circle identified by the interval $[0, 1)$, and $T'$ maps $\zeta$ to $\zeta + \alpha$, where $\alpha < 1/3$ is an irrational number. The invariant measure $\mu'$ is the usual Lebesgue measure. There is a distinguished set $I$ which is the union of the intervals $[0, \gamma)$ and $[\alpha + \beta, \alpha + \beta + \gamma)$, where $\beta, \gamma$ are constants depending on $\alpha$. The dynamics of the skew product $Y' = (\mathbb{Z}_2 \rtimes X', T')$ maps $(b, \zeta)$ to $(b + \mathbf{1}_I(\zeta), \zeta + \alpha)$, where $\mathbf{1}_I(\zeta)$ is $1$ if $\zeta \in I$, and $0$ otherwise. This skew product is minimal, but admits precisely two distinct ergodic measures $\nu_0', \nu_1'$. 

    For each $\zeta$, we define its symbolic trace $x = \mathsf{tr}(\zeta)$ as $x(n) = \mathbf{1}_I(\zeta + n\alpha)$. We claim that $\mathsf{tr}(\zeta)$ generates a minimal shift. To that end, for each $u \in \{0, 1\}^+$, we shall inductively define sets $I_u$ with the property that $u$ occurs at index $n$ in $\mathsf{tr}(\zeta)$ if and only if $\zeta + n\alpha \in I_u$. Clearly, $I_1 = I$, and $I_0 = [0, 1) \setminus I_1$, and $I_{bu} = I_b \cap (-\alpha + I_u)$. By examining the orientation of the intervals, we have that each $I_u$ is either empty, or contains a non-empty open set. Thus, $\mathsf{tr}(\zeta)$ is uniformly recurrent, and furthermore $\mathcal{L}(\mathsf{tr}(\zeta))$ is the set of factors $u$ such that $I_u$ is non-empty (we take $I_\varepsilon$ to be $[0, 1)$). It follows that every $\mathsf{tr}(\zeta)$ generates the same shift $X$, and that $f_{\mathsf{tr}(\zeta)}(u) = \mu'(I_u)$. The unique ergodicity of $X$ follows readily from the following lemma.

\begin{lemma}
    \label{lem:symbolic-UE}
        Consider the partition of the unit circle into finitely many semi-open intervals $I_j = [\beta_j, \gamma_j)$, and let $\alpha$ be irrational. The shift $X$ generated by the codings (with respect to the intervals $I_j$) of the rotation by $\alpha$ is uniquely ergodic.
    \end{lemma}
    \begin{proof}
        Denote by $\tilde I_j$ the interval $(\beta_j, \gamma_j]$.
        By symmetry, the set of codings with respect to the partition of ``dual'' intervals is also included in $X$. We shall prove that $X$ comprises precisely of the codings with respect to the original partition, and codings with respect to the dual partition. 

        For each $u \in \mathcal{L}_n(X)$, we define $I_u$ and $\tilde I_u$ as before, and also define the non-empty closed interval $\bar I_u$ as their union. If $x \in X$ is a coding of the rotation starting at $\zeta$, then for each prefix $u$ of $x$, we must have $\zeta \in \bar I_u$. Conversely, if $\zeta$ is in each $\bar I_u$, then $x$ is its coding because by the orientation of the intervals, $\zeta$ is either in each $I_u$, or in each $\tilde I_u$ (if $\zeta$ was on the left endpoint of $I_u$ and right endpoint of $\tilde I_{uv}$, it would imply that $I_{uv}$ is empty, a contradiction).

        Thus, each $x \in X$ defines an infinite sequence of nested closed intervals, and there necessarily exists a point $\zeta$ in their countable intersection. By the preceding discussion $x$ codes a rotation of this $\zeta$ with respect to either the original partition, or the dual partition. This proves that each $x \in X$ admits factor frequencies, with $f_x(u)$ given by the Lebesgue measure of $I_u$. We use Oxtoby's theorem to conclude that $x$ is uniquely ergodic. 
    \end{proof}

    We now define the shift $Y$ to be the set of runs of $\mathcal{A}$ (starting in either state) on words of $X$. Alternately, $Y$ is the minimal shift generated by symbolic traces of the minimal system $Y' = Z_2 \rtimes X'$. 

    Finally, we show that $Y$ admits two distinct invariant measures $\nu_0, \nu_1$. For this, observe that $\mathsf{tr}\colon Y' \rightarrow Y$ is measurable. Indeed, the subsets $B$ of $Y$ for which $\mathsf{tr}^{-1}(B)$ is Lebesgue-measurable form a $\sigma$-algebra which includes the cylinders of $Y$ as they are images of measurable sets $\{b\} \times I_u$. Thus, the ergodic measures $\nu_0', \nu_1'$ of $Y'$ project to invariant measures $\nu_b = \nu_b' \circ \mathsf{tr}^{-1}$. 

    To complete the proof, we record that $\nu_0, \nu_1$ are indeed distinct because $\nu_0', \nu_1'$ are such that $\nu'_b(b \times B) = \mu'(B)$, i.e., each measure concentrates its mass in one copy of $X'$ \cite[Prop.~3]{chaika-counter}.

   We thus have that $Y$ is a shift that is not uniquely ergodic, and contains the runs of $\mathcal{A}$ on all $x \in X$. In particular, by Lem.~\ref{lem:UE-equivalence}, there exists some counterexample $y = \mathcal{A}(x)$ that does not admit factor frequencies despite $x$ doing so.
\end{proof}

The natural questions that follow are: can  one identify when a shift is uniquely ergodic? Is there a stronger condition on $X$ which precludes the misbehaviour of $Y$ described above? Boshernitzan \cite{Boshernitzan1992} gave a sufficient (but not necessary) condition in answer to the former. In this paper, we shall show that this condition answers the latter in the affirmative. 

\begin{definition}[Condition (B)]
\label{def:Boshernitzan}
    Let $X \subseteq \Sigma^\omega$ be a minimal shift space. For an invariant measure $\mu$, define $\kappa\colon \mathbb{N} \rightarrow \mathbb{R}$ as $$\kappa(n) = \min_{u \in \mathcal{L}_n(X)} \left( n \cdot \mu\left([u]_X\right) \right).$$ The minimal shift space $X \subseteq \Sigma^\omega$ satisfies Boshernitzan's condition (also referred to as Condition (B)) if it admits an invariant measure $\mu$ such that the corresponding $\kappa$ satisfies $\limsup_{n \rightarrow \infty} \kappa(n) > 0$. A uniformly recurrent word $x \in \Sigma^\omega$ satisfies Boshernitzan's condition if it generates a shift space $X$ that satisfies Boshernitzan's condition.
\end{definition}

\begin{lemma}[Boshernitzan \cite{Boshernitzan1992}]
    \label{lem:Boshernitzan}
    A minimal shift space $X \subseteq \Sigma^\omega$ that satisfies Boshernitzan's condition is uniquely ergodic. If a word $x \in \Sigma^\omega$ satisfies Boshernitzan's condition, then it admits uniform factor frequencies.
\end{lemma}

\begin{remark}
\label{remark:Boshernitzan}
    If a word $x$ is linearly recurrent, then it satisfies Boshernitzan's condition \cite[Thm.~15]{durand2008}. In fact,
     linear recurrence for a minimal shift space  is equivalent to $ \liminf \kappa(n)  >0$  by  a result due to Boshernitzan (see  \cite[Exercise 174]{Fog02}). By Rmk.~\ref{remark:recurrentsuffix} below, Boshernitzan's condition is also satisfied by primitive  morphic words. The work of Damanik and Lenz \cite{damanik2006} is important from the perspective of applicability: they show that Sturmian words satisfy Boshernitzan's condition, as do almost all words coding the orbits of interval-exchange transformations \cite[Thm.~5]{damanik-applications}, a large class of 1-dimensional toric words \cite[Sec.~5]{damanik-applications}, and almost all Arnoux-Rauzy words \cite[Thm.~12]{damanik-applications}.
\end{remark}

\section{The Krohn-Rhodes Framework}

In this paper, we shall work with deterministic automata and transducers. Henceforth, when we refer to automata and transducers, we mean the deterministic models of computation described below. 

A transducer $\mathcal{A}$ is given by $(Q, q_\mathsf{init}. \Sigma, \Gamma,  \delta^o, \delta^t)$, where $Q$ is a finite non-empty set of states, $q_\mathsf{init} \in Q$ is the initial state, $\Sigma$ is the input alphabet, $\Gamma$ is the output alphabet, $\delta^t\colon Q \times \Sigma \rightarrow Q$ is the transition function, and $\delta^o\colon Q \times \Sigma \rightarrow \Gamma^*$ is the output function. For convenience, we shall extend the domain of $\delta^t,\delta^o$ to $Q \times \Sigma^*$ in the obvious way: $\delta^t(q, \varepsilon) = q$ and $\delta^o(q, \varepsilon) = \varepsilon$ for all $q$, while $\delta^t(q, ua) =\delta^t(\delta^t(q, u), a)$ and $\delta^o(q, ua) =\delta^o(q, u) \cdot\delta^o(\delta^t(q, u), a)$. A \emph{uniform} transducer is one where we have that for all $q, a$, we have $|\delta^o(q, a)|$ is the same: if the common length is $l$, we have an $l$-uniform transducer. A \emph{non-erasing} transducer is one where for all $q \in Q, a \in \Sigma$ we have $\delta^o(q, a) \ne \varepsilon$.

Automata are special cases of $1$-uniform transducers where $\Gamma = Q \times \Sigma$, and $\delta^o$ is the identity function, and we therefore present them simply as $(Q, q_\mathsf{init}, \Sigma, \delta)$. We shall denote the output of a transducer $\mathcal{A}$ upon reading a word $x$ (finite or infinite) as $\mathcal{A}(x)$. When $\mathcal{A}$ is an automaton, we call the sequence $\mathcal{A}(x)$ of state-letter pairs the \emph{run} of $\mathcal{A}$ on $x$, because indeed if we denote $\mathcal{A}(x)(n) = (q_n, a_n)$, then for every $n$ we have $\delta(q_n, a_n) = q_{n+1}$.

Given a transducer $\mathcal{A}$, we shall denote its \emph{underlying automaton} (obtained as described above) by $\mathcal{A}_0$. We observe that we can always define a substitution $\sigma\colon (Q \times \Sigma) \rightarrow \Gamma^*$ such that for all $x$, we have that $\mathcal{A}(x) = \sigma\left(\mathcal{A}_0(x)\right)$.  The following lemma assures us that it suffices to focus only on automata to prove our results.
\begin{lemma}
    \label{lem:automata-suffice}
    Let $\sigma\colon \Sigma \rightarrow \Gamma^*$ be a substitution, and let $x \in \Sigma^\omega$ be an infinite word. If $x$ has one of the following properties, then $\sigma(x)$ also has that property, provided it is an infinite word.
    \begin{enumerate}
        \item primitive morphic
        \item recurrent
        \item (effectively) uniformly recurrent
        \item (effectively) linearly recurrent
        \item uniformly recurrent and admits  factor frequencies
        \item satisfies Boshernitzan's condition
    \end{enumerate}
    In item (5), if $x$ is effectively uniformly recurrent and admits computable factor frequencies, then $\sigma(x)$ also has these properties.
\end{lemma}
\begin{proof}
    Item (1) follows immediately by definition (see Rmk.~\ref{remark:cobham}). 
    
    Items (2) through (4) have similar proofs.  To formalise the proof, consider the \emph{attribution function} $\alpha$ from indices of $\sigma(x)$ to indices of $x$. We define $\alpha(j)$ to be the smallest index $i$ such that $j < |\sigma(x(0, i+1))|$, i.e., the letter at index $j$ of $\sigma(x)$ is attributed to the image of $x(i)$. We define $\alpha(J) = \bigcup_{j \in J} \alpha(j)$. In particular, when $v$ is a factor of $\sigma(x)$ at index $j$, then we set $J = [j, j + |v|)$, obtain $[i, i')$ as the smallest interval that contains $\alpha(J)$, and say that the above occurrence of $v$ is attributed to an occurrence of $u = x(i, i')$ in $x$ at index $i$. 

    Clearly, if an occurrence of $v$ is attributed to an occurrence of $u$, then the recurrence of $u$ implies that of $v$. This proves item (2). 

    In case $x$ is uniformly recurrent, an occurrence of $v$ will always be attributed to a factor $u$ that is at most $R_x(1) \cdot |v|$ in length. This is because the gaps between occurrences of non-erasing letters (letters $a$ for which $\sigma(a)$ is not the empty word) are bounded by $R_x(1)$. The gaps between the occurrences of $u$ are themselves bounded by $R_x(R_x(1) \cdot |v|)$. Finally, the gaps between the occurrences of $v$ must be bounded by $\left(\max_{a\in \Sigma} |\sigma(a)|\right) \cdot  R_x(R_x(1) \cdot |v|)$. This proves item (3), and also item (4): if $R_x(n)$ is linear, then the bound on the gaps between occurrences of $v$ will also be linear in $|v|$.

    We prove item (5) in two steps. As in Rmk.~\ref{remark:cobham}, we decompose $\sigma$ as $\tau \circ \hat\sigma$, where $\tau$ is a letter-to-letter substitution, and $\hat\sigma$ maps each letter $a$ to $|\sigma(a)|$ copies of $a$, i.e., $\hat\sigma(a) = (a, 0)(a, 1)\cdots(a, |\sigma(a)|-1)$. We denote $\hat\sigma(x)$ by $z$, and $\tau(z)$, which is the same as $\sigma(x)$, by $y$. Observe by the preceding discussion that $z, y$ are uniformly recurrent. We shall show that $z$ admits factor frequencies, and in turn so does $y$. 
    
    By uniform recurrence, we have that any factor $w$ of $z$ can be attributed to one of finitely many factors $\{u_1, \ldots, u_d\}$ of $x$, and this factor will have length at most $R_x(1)\cdot |w|$. By the definition of attribution, there is no pair $u_i, u_j$ such that $u_i$ is a factor of $u_j$ (in particular, the cylinders $[u_1]_X, \ldots, [u_d]_X$ are disjoint); moreover by the positional-encoding nature of the construction, each $u_i$, upon restricting to non-erasing letters, will give a word $u$ such that $\hat\sigma(u)$ contains exactly one occurrence of $w$ (and conversely, the set of factors which begin and end with a non-erasing letter, and give $u$ upon restricting to non-erasing letters is $\{u_1, \ldots, u_d\}$). We thus obtain $f_z(w) = f_z(\hat\sigma(u))$. We claim that $f_z(\hat\sigma(u)) = F/L$, where $F = \sum_{i=1}^d f_x(u_i)$, and $L$ is the non-zero constant $\sum_{a \in \Sigma} f_x(a)|\hat\sigma(a)|$. Indeed, $L \ge 1/R_x(1)$, and is hence non-zero.   

     We note that $F = \lim_{N\rightarrow \infty} \frac{1}{N}\sum_{i=0}^{N-1}\sum_{j=1}^d\mathbf{1}_{u_j}(T^i x)$ by definition. By Lem.~\ref{lem:generic}, we also have that $L = \lim_{N \rightarrow \infty} \frac{1}{N}\sum_{i=0}^{N-1}|\sigma(x(i))|$ (consider the function that maps $x$ to $|\sigma(x(0))|$). We rephrase these observations by taking appropriate subsequences of Birkhoff averages as follows. Let $(\beta_N)_{N \in \mathbb{N}}$ denote the strictly increasing sequence whose $N$-th element $\beta_N$ is the index (in $x$) of the $N$-th occurrence of an element of $\{u_1, \ldots, u_d\}$. We have in particular that $\lim_{N\rightarrow \infty} \frac{N}{\beta_N} = F$. Similarly, we define the strictly increasing sequence $(\gamma_N)_{N \in \mathbb{N}}$ of indices of occurrences of $\hat\sigma(u)$ in $z$ as $\gamma_N = \sum_{i=0}^{\beta_N-1}|\sigma(x(i))|$, and observe that $\lim_{N \rightarrow \infty} \frac{\gamma_N}{\beta_N} = L$. Since $L \ne 0$, we obtain that $\lim_{N \rightarrow \infty} \frac{N}{\gamma_N} = \frac{F}{L}$. Denoting $U = \{\gamma_N \mid N \in \mathbb{N}\}$, it remains to observe (by elementary analytic means) that $\lim_{M \rightarrow \infty}\frac{|U \cap [0, M)|}{M} = \lim_{N \rightarrow \infty} \frac{N}{\gamma_N} = \frac{F}{L}$.

    Indeed, the sequences agree whenever $M = \gamma_N$. The sequence on the left decreases in the interim, and only increases when $M$ steps from $\gamma_N$ to $\gamma_N+1$. A simple calculation shows that the increment is bounded by $1/\gamma_N$, i.e., asymptotically vanishes. It then follows that the sequences converge to the same limiting value $F/L$.

    The proof that $y = \tau(z)$ inherits the property of admitting factor frequencies from $z$ is simpler because $\tau$ is letter-to-letter. For any factor $v$, there are finitely many factors $w_1, \ldots, w_d$ of the same length such that $\tau(w_i) = v$. The frequency $f_y(v)$ is simply $\sum_i f_z(w_i)$.

    To prove item (6)\footnote{The proof of \cite[Thm.~8(a)]{damanik2006}, which states item (6), implicitly assumes the substitution is non-erasing.}, observe that $$\min_{v \in \mathcal{L}_n(y)}f_y(v) \ge \frac{1}{L} \min_{u \in \mathcal{L}_{n \cdot R}(x)}f_x(u),$$ where $R = R_x(1)$. Writing $\kappa'(n)$ as shorthand for $\min_{u \in \mathcal{L}}f(u)$, using $X, Y$ to denote the shifts generated by $x, y$, and defining $\mu_X, \mu_Y$ using $f_x, f_y$, we get $L(R+1)n\cdot \kappa_Y'(n) \ge \max_{0\le r < R} (nR + r) \cdot \kappa'_X(nR+r)$. This allows us to conclude that $\limsup_{n \rightarrow \infty} \kappa_Y(n) > 0$, given $\limsup_{n \rightarrow \infty} \kappa_X(n) > 0$, thus establishing the preservation of Boshernitzan's condition.
\end{proof}

We record an interesting corollary of the above proof.

\begin{corollary}
\label{cor:non-erasing-substitution}
    Let $\sigma$ be a non-erasing substitution. If $x$ admits (computable) factor frequencies, then so does $\sigma(x)$. 
\end{corollary}

Our main tool to prove preservation results for automata will be the Krohn-Rhodes theorem, which decomposes automata into simple, ``well-behaved'' components that are connected ``in series''. The strategy is then to prove preservation results for these simpler automata. The Krohn-Rhodes theorem is often stated for semigroups \cite{krohn-rhodes-OG}; see \cite[Thm.~1]{meyer1969remarks} for a formulation in terms of automata (see also \cite[Thm.~3]{Maler2010}). We now introduce the terminology required to state the theorem. 

Given an automaton $\mathcal{A} = (Q, q_\mathsf{init}, \Sigma, \delta)$, for every $u \in \Sigma^*$, we define $\delta_u\colon Q \rightarrow Q$ to be the transition function induced by $u$, i.e., $\delta_u(q) = \delta(q, u)$. An automaton $\mathcal{A}$ is called a \emph{reset automaton} if for every $a \in \Sigma$, the function $\delta_a$ is either the identity function, or constant-valued. Intuitively, a reset automaton is a flip-flop whose memory records the last reset. An automaton $\mathcal{A}$ is called a \emph{permutation automaton} (also a \emph{group automaton}) if for every $a\in \Sigma$, $\delta_a$ defines a permutation of $Q$, i.e., is bijective. The automaton $\mathcal{A}$ is called \emph{counter-free} if there does not exist $u \in \Sigma^+, Q' \subseteq Q$ such that $\delta_u$ induces a non-trivial permutation on $Q'$. For the sake of brevity, we declare that permutation, reset, and counter-free transducers are respectively those whose underlying automata are permutation, reset, and counter-free automata. We remark that counter-free automata are of particular interest because of their equivalence with first-order-definable (aperiodic) languages, as established by the McNaughton--Papert theorem \cite{mcnaughton-papert}.

We say that an automaton $\mathcal{A}' = (Q', q_\mathsf{init}' \Sigma, \delta')$ \emph{covers} an automaton $\mathcal{A} = (Q, q_\mathsf{init}, \Sigma, \delta)$ if there exists a map $\varphi$ from $Q'$ to $Q$ that respects the initial states and commutes with the transition relations, i.e., $\varphi(q_\mathsf{init}') = \mathsf{q}_\mathsf{init}$ and for every $q' \in Q', a \in \Sigma$, we have that $\varphi(\delta'(q', a)) = \delta(\varphi(q'), a)$. This implies, in particular that for a word $u$ over $\Sigma$, its runs $v, v'$ over $\mathcal{A}, \mathcal{A}'$ respectively are related by $v = \sigma'(v')$ where $\sigma'$ is a letter-to-letter substitution.

If $\mathcal{A}_1 = (Q_1, \Sigma, \delta_1)$, and $\mathcal{A}_2 = (Q_2, Q_1 \times \Sigma, \delta_2)$, then the \emph{cascade} $\mathcal{A}_2 \circ \mathcal{A}_1$ is an automaton over $\Sigma$ with states $Q_2 \times Q_1$ defined by the following property. The run of a word $u$ is a word $v'$ over $(Q_2 \times Q_1 \times \Sigma)$: its projection $v$ onto $Q_1 \times \Sigma$ is the run of $\mathcal{A}_1$ on $u$, and $v'$ is the run of $\mathcal{A}_2$ on $v$. Formally, the transition function is given by 
\[\delta((q_2, q_1), a) = (\delta_1(q_1, a), \delta_2(q_2, (q_1, a))). \]
The cascade $\mathcal{A}_k \circ \cdots \circ \mathcal{A}_1$ is implemented by performing the rightmost cascade first, akin to function composition.

\begin{theorem}[Krohn-Rhodes \cite{meyer1969remarks}]
\label{thm:krohn-rhodes}
    For every automaton $\mathcal{A}$, we can compute a cascade $\mathcal{A}' = \mathcal{B}_k \circ \cdots \circ \mathcal{B}_1$ such that:
    \begin{enumerate}
        \item $\mathcal{A}'$ covers $\mathcal{A}$.
        \item Each $\mathcal{B}_i$ is either a permutation automaton or a two-state reset automaton.
        \item If $\mathcal{B}_i$ is a permutation automaton, then its transition group is homomorphic to a subgroup of the transition monoid of $\mathcal{A}$.
    \end{enumerate}
    In particular, if $\mathcal{A}$ is a counter-free automaton, then each $\mathcal{B}_i$ is a reset automaton.
\end{theorem}

\begin{remark}
    \label{remark:template}
    Our proofs of preservation theorems for transducer outputs will follow the following template: (i) prove the theorem for reset automata and permutation automata; (ii) deduce that the theorem holds for any cascade of these special automata; (iii) apply Lem.~\ref{lem:automata-suffice} (we only need the easy case of $\sigma$ being a coding here) and deduce that the theorem holds for any automaton covered by such a cascade; (iv) apply the Krohn-Rhodes theorem and deduce that the theorem indeed holds for all automata; (v) apply Lem.~\ref{lem:automata-suffice} to deduce the theorem holds for all transducers. In the above, only Step (i) will require work.
\end{remark}

\section{Preservation of Recurrence}
\label{sec::preservation-of-recurrence}
Sem\"enov showed that if a word $x \in \Sigma^\omega$ is \emph{effectively almost-periodic} (i.e., the word $x$ is effective, and given any $u \in \Sigma^+$, it can be decided whether $u$ occurs infinitely often in $x$, and if it is the case, we can compute ${\mathcal R}_x (u)$), then for any transducer $\mathcal{A}$, the word $\mathcal{A}(x)$, if infinite, is also effectively almost-periodic (see \cite[Sec.~3]{semenov-english} for an exposition of that result in English). Using this result and techniques from the proof, Pritykin showed that if a word $y \in \Sigma^\omega$ has a uniformly recurrent suffix $x$, then for any transducer $\mathcal{A}$, the word $\mathcal{A}(y)$, if infinite, also has a uniformly recurrent suffix \cite{pritykin}. We use the Krohn-Rhodes theorem to prove (in an arguably simpler and more insightful manner) the following strengthening. We remark that the Krohn-Rhodes theorem is crucially needed in the case where we do not assume uniform recurrence.

\begin{theorem}
\label{thm:recurrence}
    Let $\mathcal{A} = (Q, q_\mathsf{init}, \Sigma, \Gamma, \delta^o, \delta^t)$ be a transducer. Consider a word $x \in \Sigma^\omega$. If $x$ has one of the following properties, then $\mathcal{A}(x)$ also has the same property, provided that it is an infinite word.
    \begin{enumerate}
        \item (effectively) recurrent suffix
        \item (effectively) uniformly recurrent suffix (whose starting index is computable) 
        \item (effectively) linearly recurrent suffix (whose starting index is computable)
    \end{enumerate}
\end{theorem}

 By Rmk.~\ref{remark:template}, the preservation theorem above would follow from preservation lemmas for permutation automata and reset automata. Observe that in the proofs, we can assume without losing generality that the input word $x$ is recurrent by simply considering the run starting at the beginning of the recurrent suffix. The following distills the core idea of Sem\"enov's argument.

\begin{lemma}
\label{lem:recurrence-perm}
    Let $\mathcal{A} = (Q, q_\mathsf{init}, \Sigma, \delta)$ be a permutation automaton. Consider a word $y \in \Sigma^\omega$. If $x$ has one of the following properties, then $\mathcal{A}(x)$ also has that property.
    \begin{enumerate}
        \item (effectively) recurrent
        \item (effectively) uniformly recurrent
        \item (effectively) linearly recurrent
    \end{enumerate}
\end{lemma}
\begin{proof}
    We prove the preservation of recurrence, and obtain preservation of the other properties as corollaries of the proof. We shall consider an arbitrary occurrence of a factor $v$ of $\mathcal{A}(x)$, and prove that $v$ occurs infinitely often by showing that we can find another occurrence to the right of the one under consideration. 
    
    To that end, let us establish some notation. We denote the projection of $v =v_0$ onto $\Sigma$ by $u_0$, and denote by $i_0$ the index of the occurrence of $v_0$ in $\mathcal{A}(x)$ (which is the same as that of the corresponding occurrence of $u_0$ in $x$). If the first letter of $v_0$ is $(q, u_0(0))$ and $\delta_{u_0}(q) = q_0$, then we use $(u_0, q_0)$ as shorthand for $v_0$. The shorthand indeed defines $v_0$ because $\mathcal{A}$ is a permutation automaton, and hence for every word $u$, $\delta_u$ admits an inverse, which we shall denote by $\delta_u^{-1}$. In other words, the transition function is ``reversible'': if we know the state after reading $u$, we can determine all the states along the run of $u$.

    We inductively define sequences $(u_k)_k$ of factors of $x$, $(v_k)_k$ of factors of $\mathcal{A}(x)$, $(i_k)_k$ of indices, and $(q_k)_k$ of states. We shall maintain that the projection of each $v_k$ onto $\Sigma$ is the corresponding $u_k$. 
    We always have that $u_k$ has infinitely many occurrences in $x$, and we will maintain that $u_k$ occurs in $x$ at index $i_0$. We find its next occurrence at index $i_{k+1} > i_0$, and define $u_{k+1} = x(i_0, i_{k+1} + |u_k|)$ to be the \emph{extended return word} that spans consecutive occurrences of $u_k$. Correspondingly, we define $v_{k+1}= \mathcal{A}(x)(i_0, i_{k+1}+|u_k|)$. This word $v_{k+1}$ will be denoted by $(u_{k+1}, q_{k+1})$. In particular, the word $u_k$ is a prefix  of $u_{k+1}$.

    By the pigeonhole principle, there exist distinct $l, k$ with $0 \le l < k \le |Q|$ such that $q_l = q_k$. By construction, we also have that $u_l$ is both a prefix and a suffix of $u_k$. This means that we have found two occurrences of $v_l = (u_l, q_l)$, one at index $i_0$, and the other at index $i_{k+1} + |u_k| - |u_l|$. Our construction also ensures in particular that $u_0$ is a prefix of $u_l$, and since the transition function is reversible, we have found an occurrence of $v_0$ at index $i_{k+1} + |u_k| - |u_l|$, which is to the right of $i_0$ by at most $|u_{k+1}|$.

    We have thus far proven that $\mathcal{A}(x)$ is recurrent. Effectiveness follows because $\mathcal{A}(x)$ inherits the property of having a decidable MSO theory from $x$. 
    
    In the case of (effective) uniform recurrence, we observe that $i_{k+1} - i_0 \le R_x(|u_k|)$, and thus $|u_{k+1}| \le R_x(|u_k|) + |u_k|$. For convenience, we define the auxiliary function $W_x$ (read as ``window function'') as $W_x(n) = R_x(n) + n$. From the discussion above, we get that $R_{\mathcal{A}(x)}(n) \le W_x^{(|Q|+1)}(n)$, where $W_x^{(k)}$ denotes the $k$-fold composition of $W_x$. Clearly, if $R_x$ is computable (respectively, linear), then so is $R_{\mathcal{A}(x)}$. 

    Finally, we check that in the case of uniform recurrence, the index of the first occurrence of $v$ is bounded by $R_{\mathcal{A}(x)}(|v|)$. This is indeed the case, because if $i_0$ exceeds the above bound, we can apply the above argument \emph{mutatis mutandis} to find an occurrence of $v$ to the \emph{left} of $i_0$. This completes the proof.
\end{proof}

\begin{lemma}
\label{lem:reset}
    Let $\mathcal{A} = (Q, q_\mathsf{init}, \Sigma, \delta)$ be a reset automaton. Consider a word $x \in \Sigma^\omega$. If $x$ has one of the following properties, then $\mathcal{A}(x)$ also has that property.
    \begin{enumerate}
        \item recurrent suffix 
        \item effectively recurrent suffix with computable starting index
        \item (effectively) uniformly recurrent suffix (whose starting index is computable)
        \item (effectively) linearly recurrent suffix (whose starting index is computable)
        \item uniformly recurrent suffix and admits factor frequencies
    \end{enumerate}
    In item (5), if the suffix of $x$ is effectively uniformly recurrent, its starting index is computable, and $x$ admits computable factor frequencies, then $\mathcal{A}(x)$ also has these properties.
\end{lemma}

\begin{proof}
We shall assume without losing generality that $x$ is recurrent by simply considering the run starting at the beginning of its recurrent suffix. 
The lemma is trivial if $x$ does not contain any recurrent reset letters, i.e., letters $a$ such that $\delta_a$ is a constant-valued function. We shall therefore assume that $x$ has a recurrent reset letter $a$, whose first occurrence is at index $N$, which can be computed in the case of effective uniform recurrence. We claim that the suffix of $\mathcal{A}(x)$ starting at index $N+1$ is recurrent.

The key will be to focus on return words to the reset letter $a$. By definition, these words begin with $a$, and have exactly a single occurrence of $a$ because $a$ is a single-letter word. We have that $x(N, \infty) = r_0r_1r_2\cdots$, a concatenation of return words. 

As before, we shall show that for any occurrence of an arbitrary factor $v$ at index $i_0 > N$, we can find another occurrence at an index to the right of $i_0$. Let $u$ be the projection of $v$ onto $\Sigma^+$, i.e., $x(i_0, i_0+|v|) = u$. We observe by the above factorisation of $x(N, \infty)$ into return words that there exists an index $i$ with $N \le i < i_0$ such that $x(i) = a$. We take $i$ to be maximal, and in particular we have $i_0 - i < R_x(1)$. Let $u'$ denote the factor $x(i, i_0+|v|)$. By recurrence, $u'$ will have another occurrence at an index $i' > i$, and this will lead to an occurrence of $v$ in $\mathcal{A}(x)$ at index $i' - i + i_0$. This proves that $y = \mathcal{A}(x)(N+1, \infty)$ is recurrent, establishing item (1).

To prove item (2), observe that when $x$ has a decidable MSO theory, we can in particular decide if a reset letter exists, and if so, find its first occurrence by brute enumeration. We also have that the corresponding suffix of $\mathcal{A}(x)$ has a decidable MSO theory, and is hence effectively recurrent.

If $x$ is uniformly recurrent, then $|i' - i| \le R_x(|v| + R_x(1))$, and hence $R_{y}(n) \le R_x(n + R_x(1))$. Clearly $R_{y}$ is effective (respectively, linear) if $R_x$ is. We also check that if $i_0 - N$ exceeds $R_x(n + R_x(1))$, then we can use the same argument as above to find an occurrence of $v$ in $\mathcal{A}(x)$ to the left of $i_0$. This proves items (3) and (4).

To study factor frequencies of $v$, we partition the occurrences of $u$, and account precisely for which partitions result in an occurrence of $v$. Formally, we construct a tree whose root is $u$, vertices are factors $u'$ of the form $wu$, leaves are factors $u'$ that have $u$ as a proper suffix and begin with a reset letter, and the successors of the internal nodes $u'$ are of the form $au'$. The depth of the tree is thus at most $R_x(1)$. The frequency $f_y(v)$ is simply $\sum_{u'}f_x(u')$, where the summation ranges over leaves $u'$ such that the run of the reset automaton on the word $u'$ has $v$ as its suffix.
\end{proof}

\begin{remark}
\label{remark:reset-boshernitzan}
    In the above proof, consider the case $x$ is uniformly recurrent, and let $y$ be the recurrent suffix of $\mathcal{A}(x)$. Observe that $\min_{v \in \mathcal{L}_n(y)} f_y(v) \ge \min_{u \in \mathcal{L}_{n+R_x(1)}(x)}f_x(u)$. 
\end{remark}

Lem.~\ref{lem:recurrence-perm}, Lem.~\ref{lem:reset}, and Rmk.~\ref{remark:template} prove Thm.~\ref{thm:recurrence}. In fact, item (4) of Lem.~\ref{lem:reset}, along with the counter-free case of Thm.~\ref{thm:krohn-rhodes}, gives the following result.

\begin{theorem}
    \label{thm:counter-free-freq}
    Let $\mathcal{A} = (Q, q_\mathsf{init}, \Sigma, \Gamma, \delta^o, \delta^t)$ be a counter-free transducer. Let $x \in \Sigma^\omega$ be uniformly recurrent and admit factor frequencies. We have that if $\mathcal{A}(x)$ is an infinite word, then it has a suffix $y$ which is uniformly recurrent and admits factor frequencies. When $x$ is effectively uniformly recurrent and admits computable factor frequencies, the suffix $y$ also has these properties, and its starting index in $\mathcal{A}(x)$ can be computed.
\end{theorem}

\section{Preservation of Self-Similarity}

In this section, we show that transducers preserve properties that indicate self-similarity\footnote{Every word $x = a_0a_1a_2 \cdots$ is vacuously $S$-adic for $S = \{\sigma_1, \ldots, \sigma_{|\Sigma|}, \tau_1, \ldots, \tau_{|\Sigma|}\}$ over the alphabet $\Sigma \cup \{b\}$, where $\sigma_i$ replaces the distinguished letter $b$ with $a_i \in \Sigma$ and is identical elsewhere, and $\tau_i$ replaces $b$ with $ba_i$ and is identical elsewhere. We then see that $x$ is directed by $\sigma_{a_0}\tau_{a_1}\tau_{a_2}\cdots$. The point of directive sequences is to encode more information about self-similar structure, and the lemma conveys the sense in which this structure is preserved. In particular, it allows us to prove that if $x$ is morphic then so is $\mathcal{A}(x)$.}.

\begin{lemma}
    \label{lem:self-similar}
    Let $\mathcal{A} = (Q, q_\mathsf{init}, \Sigma, \Gamma, \delta^o, \delta^t)$ be a transducer, let $S$ be a set of substitutions, let $x \in \Sigma^\omega$ be an $S$-adic word directed by the sequence $(\sigma_n)_{n=0}^\infty$. We have that the word $\mathcal{A}(x)$, if infinite, is $\hat S$-adic for a set $\hat S$ of substitutions defined using only $S$, and the directive sequence $(\hat \sigma_n)_{n=0}^\infty$ can be defined using only $\mathcal{A}$ and $(\sigma_n)_{n=0}^\infty$.
\end{lemma}

\begin{proof}
    We can assume without losing generality that $\mathcal{A} = \mathcal{A}_0$ is an automaton (Lem.~\ref{lem:automata-suffice}). Let $x_0 = x, x_1, \ldots$ be the sequence of words such that for all $n$, $x_n \in \Sigma_n^\omega$ and $x_n = \sigma_n(x_{n+1})$. To show that the run $y_0 = \mathcal{A}_0(x_0)$ is $\hat S$-adic, we shall construct a sequence $(\mathcal{A}_n)_{n=0}^\infty$ of automata, where $\mathcal{A}_n = (Q, q_\mathsf{init}, \Sigma_n, \delta_n)$, then let $y_n = \mathcal{A}_n(x_n)$, and finally define a sequence of substitutions $(\hat\sigma_n)_{n=0}^\infty$ such that for all $n$, $y_n = \hat\sigma_n(y_{n+1})$.

    The key idea is that $x_{n+1}$ can be intuited as a ``compression'' of $x_n$ using $\sigma_n$.
    Hence, in $\mathcal{A}_{n+1}$, we define $\delta_{n+1}$ such that for all $a \in \Sigma_{n+1}$, $\delta_{n+1}(q, a) = \delta_n(q, \sigma_n(a))$. Thus, the run $y_{n+1}$ can be regarded as a ``fast-forwarded'' version of the run $y_n$. We are therefore motivated to define $\hat \sigma_n$ as the dual ``slow-motion'' operator, i.e., $\hat \sigma_n((q, a))$ gives the run of $\mathcal{A}_n$ on $\sigma_n(a)$ starting in state $q$. It is now easy to check that $\mathcal{A}_n(x_n) = \hat\sigma_n(\mathcal{A}_{n+1}(x_{n+1}))$ for all $n$.
\end{proof}

\begin{theorem}
    \label{thm:morphic-preservation}
    Let $\mathcal{A} = (Q, q_\mathsf{init}, \Sigma, \Gamma, \delta^o, \delta^t)$ be a transducer, and let $x \in \Sigma^\omega$ be a primitive morphic word given as the image under $\tau$ of a particular fixed point of a primitive substitution $\sigma$. When $\mathcal{A}(x)$ is an infinite word:
    \begin{enumerate}
        \item The word $\mathcal{A}(x)$ has a primitive morphic suffix $y$. We can compute the index at which $y$ begins, and compute substitutions $\hat \sigma, \hat \tau$ such that $\hat \sigma$ is primitive and $y$ is the image under $\hat \tau$ of a fixed point of $\hat \sigma$. 
        \item The word $\mathcal{A}(x)$ admits computable factor frequencies.
    \end{enumerate}
\end{theorem}
\begin{proof}
    We note that item (2) follows from item (1) by applying Lem.~\ref{lem:morphic-freq}. We therefore focus on proving item (1). We remark that it suffices to give a proof assuming $\mathcal{A}$ is an automaton. 

    We first prove that $\mathcal{A}(x)$ is morphic, and we will follow the construction given in the proof of Lem.~\ref{lem:self-similar} in order to do so. The directive sequence of $x$ is $\tau \sigma^\omega$, and in particular, for all $n \ge 1$, $\Sigma_n = \Sigma_1$ and $x_n = x_1$. Observe that this implies that the set $\{\mathcal{A}_n \mid n \in \mathbb{N}\}$ is finite and can be effectively enumerated. Furthermore, when $n \ge 1$, $\mathcal{A}_{n+1}$ and $\hat\sigma_n$ will depend only on $\mathcal{A}_n$ and $\sigma$. By the pigeonhole principle, there exist computable $m, n$ such that $m < n$ and $\mathcal{A}_m = \mathcal{A}_n$. We have that $\mathcal{A}_m(x)$ is a fixed point of $\hat\sigma_m \cdots \hat\sigma_{n-1}$, and its image under $\hat\sigma_0\cdots \hat\sigma_{m-1}$ gives $\mathcal{A}(x)$, which is hence morphic.

    Now, since $x$ is primitive morphic, it is effectively linearly recurrent (Lem.~\ref{lem:morphic-freq}). The run $\mathcal{A}(x)$ therefore has a suffix $y$ that is effectively linearly recurrent, and the starting index $N$ of this suffix can be computed (Thm.~\ref{thm:recurrence}). This suffix $y$ of the morphic word $\mathcal{A}(x)$ is again morphic, and the substitutions describing it can be computed (Rmk.~\ref{remark:recurrentsuffix}). Having proved that $y$ is an effectively linearly recurrent morphic word, we invoke Rmk.~\ref{remark:recurrentsuffix} again, this time to deduce that $y$ is indeed primitive morphic, and we can compute substitutions $\hat\sigma, \hat\tau$ such that $\hat\sigma$ is primitive, and $y$ is the image under $\hat\tau$ of a fixed point of $\hat\sigma$.
\end{proof}

For the sake of completeness, we comment on the case of \emph{automatic} words. Recall that a word $x \in \Gamma^\omega$ is automatic \cite[Chap.~5]{Allouche_Shallit_2003} if for some $k \ge 2$, there exists an automaton that  computes $x(n)$ when given the $k$-ary representation of $n$ as input: we then say that $x$ is $k$-automatic. Clearly, for $a \ge 1$, we have $x$ is $k^a$-automatic if and only if $x$ is $k$-automatic. Equivalently, by a result due to Cobham \cite[Thm.~6.3.2]{Allouche_Shallit_2003}, a word $x$ is $k$-automatic if and only if it is the image under a coding of a fixed point of a $k$-uniform morphism $\sigma$ (for every $a \in \Sigma$, $|\sigma(a)| = k$). Note that the definition can be relaxed (analogously to Rmk.~\ref{remark:cobham}) to replace the coding by a $k'$-uniform morphism \cite[Cor.~6.8.3]{Allouche_Shallit_2003}. Our proof of Lem.~\ref{lem:self-similar} can thus be adapted to deduce the main result of \cite[Sec.~6.9]{Allouche_Shallit_2003}: 

\begin{lemma}
\label{lem:automatic}
If $\mathcal{A}$ is a uniform transducer and $x$ is $k$-automatic, then $\mathcal{A}(x)$ is $k$-automatic.
\end{lemma}

\section{Preservation of Condition (B)}
\label{sec:rigidity}
In this section, we prove the following main result using the Krohn-Rhodes theorem. We give a general statement, but observe that by Rmk.~\ref{remark:template}, it suffices to prove the result separately for reset automata and permutation automata.

\begin{theorem}
    \label{thm:boshernitzan-preserve}
    Let $\mathcal{A} = (Q, q_\mathsf{init}, \Sigma, \Gamma, \delta^o, \delta^t)$ be a transducer, and let $x \in \Sigma^\omega$ be a/an (effectively) uniformly recurrent word that satisfies Boshernitzan's condition. We have that the word $\mathcal{A}(x)$, if infinite, has a/an (effectively) uniformly recurrent suffix $y$ that satisfies Boshernitzan's condition, and hence admits (computable) factor frequencies.
\end{theorem}

We already have argued the existence of a/an (effectively) uniformly recurrent suffix $y$ in Thm.~\ref{thm:recurrence}. We shall prove that this suffix satisfies Boshernitzan's condition. The fact that $y$ admits (computable) factor frequencies would then follow immediately from (Lem.~\ref{lem:ur-ue-compfreq} and) Lem.~\ref{lem:Boshernitzan}.

\begin{lemma}
    \label{lem:boshernitzan-reset}
    Let $\mathcal{A} = (Q, q_\mathsf{init}, \Sigma, \delta)$ be a reset automaton and let $x \in \Sigma^\omega$ be a uniformly recurrent word that satisfies Boshernitzan's condition. We have that the word $\mathcal{A}(x)$ has a uniformly recurrent suffix $y$ that satisfies Boshernitzan's condition.
\end{lemma}
\begin{proof}
    Recall from Lem.~\ref{lem:reset} that $\mathcal{A}(x)$ has a uniformly recurrent suffix $y$ that admits factor frequencies. We can define an invariant measure $\mu$ on the minimal shift $Y$ defined by $y$, as $\mu([v]_Y) = f_y(v)$ (the invariant measure on the shift defined by $x$ is similarly defined). Let us now use Rmk.~\ref{remark:reset-boshernitzan} to argue that $Y$ satisfies Boshernitzan's condition (Def.~\ref{def:Boshernitzan}) by virtue of admitting the above invariant measure. We have from Rmk.~\ref{remark:reset-boshernitzan} that for $n$ large enough, $\min_{v \in \mathcal{L}_n(y)} f_y(v)$ is lower bounded by $ \min_{u \in \mathcal{L}_{2n}(x)}f_x(u)$ as well as $\min_{u \in \mathcal{L}_{2n+1}(x)}f_x(u))$, or in other words $$3 \cdot \kappa_Y(n) \ge \max(\kappa_X(2n), \kappa_X(2n+1)).$$ This implies that $\limsup_{n\rightarrow \infty}  \kappa_Y(n) > 0$, and the word $y$ indeed satisfies Boshernitzan's condition.
\end{proof}

The case where $\mathcal{A}$ is a permutation automaton is much more involved, and we need to invoke arguments from cohomology and topological  dynamics  inspired by \cite{berthe2024density}. We devote the rest of this section to the proof of the following lemma. 

\begin{lemma}
    \label{lem:boshernitzan-perm}
    Let $\mathcal{A} = (Q, q_\mathsf{init}, \Sigma, \delta)$ be a permutation automaton. If $x$ is a uniformly recurrent word that satisfies Boshernitzan's condition, then so is $\mathcal{A}(x)$.
\end{lemma}

The group $G$ generated by the transitions of the permutation automaton, and the natural onto morphism $\varphi\colon \Sigma^* \rightarrow G$ (the morphism maps a word $u$ to the element corresponding to $\delta_u$) will play a key role in the proof. Let $X$ be the minimal shift defined by $x$, let $\mu_X$ be its unique invariant measure, and recall that by Def.~\ref{def:Boshernitzan}, $\limsup_{n \rightarrow \infty} \kappa_X(n) > 0$. We define the \emph{skew product} topological dynamical system\footnote{We acknowledge that the notation $T$ is overloaded, but nevertheless use it for clarity because, as we shall see, it performs the same shift operation in spirit. If the shift operator is invoked on multiple shift spaces within the same context, we shall use a distinguishing subscript.} $(G \rtimes X, T)$ with the compact metric state space $\{(g, x) \mid g \in G, x \in X\}$, and update $T(g, x) = (g \cdot \varphi(x(0)), x(1, \infty))$. Observe that ${T}^{-1}\circ (g, x) = \{(ga^{-1}, ax) \mid ax \in X\}$, where $a^{-1}$ is shorthand for $\varphi(a)^{-1}$.

We shall equivalently (in fact, interchangeably since we deal with a permutation automaton) regard the skew product as a shift space $Z$ over the alphabet $(G \times \Sigma)$. We map $(g, x)$ to $y$ (and also denote $y$ by $(g, x)$), where $y(n) = (g \cdot \varphi(x(0, n)), x(n))$. This perspective is isomorphic and each element $y \in Z$ can naturally be viewed as the run of an automaton with states $G$ on its projection $x \in X$. Note that $\mathcal{A}(x)$ can be obtained by simply applying a coding (letter-to-letter substitution) on $y$. By Lem.~\ref{lem:recurrence-perm}, each $y \in Z$ is uniformly recurrent (however, $Z$ is not  necessarily minimal). Given a word $x$ whose run $\mathcal{A}(x) = (e, x) = y$ we are to study, let $Y$ be the minimal shift defined by $y$: note that $Y$ is contained in (but not necessarily equal to) $Z$.  We use $\pi$ to denote the projection of $Y$ onto $X$, and $\pi^{-1}$ to denote its pre-image. 

Our goal is to prove that if $X$ satisfies Boshernitzan's condition (Def.~\ref{def:Boshernitzan}), then so does $Y$. We shall do so\footnote{Some proof ideas are inspired by \cite{berthe2024density}, especially Prop.~3.8 \emph{ibidem}.} by starting with an arbitrary invariant measure $\nu$ on $Y$, and using it to construct an invariant measure $\mu_Y$ on $Y$ that satisfies $\kappa_Y(n) = c \cdot \kappa_X(n)$ for all large enough $n$, where $c$ is a constant.

We shall study the symmetry of $Y$ in order to understand how much structure $Y$ inherits from $X$. Consider the group action of $G$ on $Z$, where $h \in G$ maps $(g, x)$ to $(hg, x)$. By the associativity of the binary group operation, we have that the action of $h$ commutes with the shift operation $T_Z$ and its inverse. Define the subgroup $H$ as $\{h \in G \mid hY = Y\}$. We shall show the following connection between invariant measures on $Y$ and $X$.

\begin{lemma}
    \label{lem:projection}
    Let $\nu$ be an arbitrary invariant measure on $Y$.
    \begin{enumerate}
        \item The map $\nu \circ h$ is an invariant measure on $Y$ for each $h \in H$. 
        \item The map $\nu \circ \pi^{-1}$ on $X$ is the same as the unique invariant measure $\mu_X$.
    \end{enumerate}
    In particular, the invariant measure $\mu_Y$ defined as $\frac{1}{|H|}\sum_{h\in H} \nu \circ h$ satisfies $\mu_Y \circ \pi^{-1} = \mu_X$.
\end{lemma}

The following lemma, which is inspired by \cite[Prop.~2.1]{coboundary-inspiration}, is instrumental in making $\pi^{-1}$ explicit: it tells us that we can express $\pi^{-1}(x) = \{(g, x) \mid g \in \alpha(x)\}$.

\begin{lemma}
    \label{lemma:coboundary}
    The map $\alpha$ given by $\alpha(x) = \{g \mid (g, x) \in Y\}$ is a well defined function from $X$ to $\mathcal{C} = H \backslash G$, satisfies $\alpha(ax) = \alpha(x) \cdot a^{-1}$, and is continuous. In particular, $X$ can be expressed as a finite union of cylinders on which $\alpha$ is constant.
\end{lemma}
\begin{proof}
    Note that $\pi(Y)$ must be a shift contained in $X$ because the projection $\pi$ commutes with $T$ and is a closed map as $G$ is compact; since $X$ is a minimal shift, $\pi(Y) = X$. This assures us that $\alpha$ maps each $x \in X$ to a non-empty set.

    We first prove that $\alpha(x)$ is indeed a coset. By the defining property of the subgroup $H$, we deduce that if $(g, x) \in Y$, then $Hg \subseteq \alpha(x)$. We now need to prove that $\alpha(x)$ cannot span more than one coset, i.e., if $(g, x) \in Y$ and $(g', x) \in Y$, then $g' = hg$ for some $h \in H$. Let $Y' = g'g^{-1}Y$, and observe that $Y'$ is also a minimal shift (recall that the group action commutes with the shift operation). We have that $(g', x)$ is contained in the intersection of the minimal shifts $Y, Y'$: this is only possible if $Y = Y'$. In other words, $h = g'g^{-1} \in H$. We henceforth refer to $\alpha$ as a function from $X$ to $\mathcal{C}$; in technical parlance it is a \emph{(minimal) cobounding map modulo $H$}.

    The next requirement follows readily. Let $x, ax \in X$. Since $\alpha(x)$ and $\alpha(ax)$ are both cosets, they have the same cardinality, and it suffices to observe from the construction of $Y$ that $\alpha(ax) \cdot \varphi(a) \subseteq \alpha(x)$ in order to establish that $\alpha(ax) = \alpha(x) \cdot a^{-1}$. 

    Thirdly, we show that $\alpha$ is continuous. Suppose for the sake of deriving a contradiction that $\alpha$ is not continuous at some $x \in X$, i.e., $\alpha(x) = C$, and for every index $n$, there exists $x_n \in X$ such that $x(0, n) = x_n(0, n)$ but $\alpha(x_n) = C_n \ne C$. By the pigeonhole principle, we deduce that there exist a coset $C'$, infinitely many indices $n$, and words $x_n$ such that $x_n(0, n) = x(0, n)$, but $\alpha(x_n) = C' \ne C$. Consider an element $g' \in C$, and observe that $((g', x_n))_n$ is an infinite sequence of elements in the compact space $Y$, and this sequence converges to $(g', x)$. However, $g' \notin \alpha(x)$, i.e., the limit does not exist in $Y$: a contradiction, as desired. 

    It remains to show that $X$ is a finite union of cylinders on which $\alpha$ is constant. From the preceding argument for continuity, it follows that for any $x \in X$, we can obtain cylinders $[u_j]_X$ on which $\alpha$ is constant and $x(j, \infty) \in [u_j]_X$. We show by contradiction that these cylinders cover $X$: if they leave some $x'$ uncovered, we can use continuity to identify a cylinder $[u']_X$ on which $\alpha$ is constant and $x' \in [u']_X$. However, by the minimality of $X$, the word $u'$ must have been a prefix of some $x(j, \infty)$: a contradiction. Finally, since $X$ is compact, the open cover $([u_j])_{j=0}^\infty$ of cylinders with the desired property admits a finite subcover $([u_j])_{j=0}^N$, which we use to establish the last claim of the lemma.
\end{proof}

\begin{proof}[Proof of Lem.~\ref{lem:projection}]
Lem.~\ref{lemma:coboundary} gives the framework to deduce that: (i) the operators $T_Y, T_Y^{-1}$ on $Y$ give the same evaluations as their counterparts $T_Z, T_Z^{-1}$; (ii) the inverse image $\pi^{-1}$ commutes with $T$ and $T^{-1}$. These will respectively imply the two parts of the claim.

Observe that $T_Y^{-1} \circ (g, x) \subseteq T_Z^{-1} \circ (g, x) = \{(ga^{-1}, ax) \mid ax \in X\}$. On the other hand, $\alpha(ax) = \alpha(x) \cdot a^{-1}$: since $g \in \alpha(x)$, this proves by the definition of $\alpha$ that the set inclusion is actually an equality, i.e., $T_Y^{-1} = T_Z^{-1}$ (that $T_Y = T_Z$ on $Y$ is obvious by the definition of $Y$). This proves (i). It then immediately follows that $h \circ T_Y^{-1} = T_Y^{-1} \circ h$, and hence $\nu \circ h$ is an invariant measure on $Y$ for each $h \in H$. Indeed, $\nu \circ h \circ T_Y^{-1} = \nu \circ T_Y^{-1} \circ h = \nu \circ h$.

To show (ii), observe furthermore that $\pi^{-1} \circ T_X^{-1}(x) = \{(g', ax) \mid g' \in \alpha(ax), ax \in X\} = \{(ga^{-1}, ax) \mid g \in \alpha(x), ax \in X\} = T_Y^{-1} \circ \pi^{-1}(x)$. This implies that for any invariant measure $\nu$, we have $\nu \circ \pi^{-1} \circ T_X^{-1} = \nu \circ T_Y^{-1} \circ \pi^{-1} = \nu \circ \pi^{-1}$. It remains to observe that $\nu \circ \pi^{-1}$ is indeed a Borel measure on $X$ because the function $\pi$ is Borel measurable by virtue of being continuous. Thus, $\nu \circ \pi^{-1}$ is an invariant measure on $X$; since $(X, T_X)$ is uniquely ergodic, it can only be the same as $\mu_X$.
\end{proof}

We are now ready to complete the proof that $Y$ inherits the property of admitting an invariant measure that satisfies Boshernitzan's condition from $X$.

Let us examine what $\mu_Y$ evaluates to on cylinders. For convenience, we shall denote a finite factor $v$ as $(g, u)$, where $v(n) = (g \cdot \varphi(u(0, n)), u(n))$ for all $n$. We have
$$
\mu_Y([(g, u)]_{Y}) = \frac{1}{|H|} \sum_{h \in H}\nu \circ h([(g, u)]_{Y}) = \frac{1}{|H|} \sum_{g' \in Hg}\nu([(g', u)]_{Y}),
$$
and by symmetry, $\mu_Y([(g, u)]_Y) = (1/|H|)\cdot \sum_{g' \in Hg}\mu_Y([(g', u)]_Y)$.

Now, by Lem.~\ref{lemma:coboundary}, for all long enough factors $u \in \mathcal{L}(X)$, we have that $\alpha$ is constant-valued on $[u]_X$. This implies that for $u$ long enough, $\pi^{-1}([u]_X) = \bigcup_{g' \in \alpha([u]_X)} [(g', u)]_{Y}$. In other words, $\mu_Y([(g, u)]_{Y}) = (1/|H|)\cdot \mu_Y \circ \pi^{-1}([u]_X)$. Since $\mu_Y \circ \pi^{-1} = \mu_X$, we have thus proven that for all long enough finite factors, we have $\mu_Y([(g, u)]_{Y}) = (1/|H|)\cdot \mu_X([u]_X)$, and it follows that $\kappa_Y(n) = (1/|H|)\kappa_X(n)$. We have thus proved that the run $y$ of an automaton with states $G$ and initial state $e$ (where $e$ is the group identity) defines a minimal shift $Y$ that satisfies Boshernitzan's condition (Def.~\ref{def:Boshernitzan}).

We obtain $\mathcal{A}(x)$ by a simple coding that substitutes $(g, a)$ with $(q, a)$, where $q$ is obtained by the action of $g$ on $q_\mathsf{init}$. By Lem.~\ref{lem:automata-suffice}, $\mathcal{A}(x)$ satisfies Boshernitzan's condition (Def.~\ref{def:Boshernitzan}), and hence by Lem.~\ref{lem:Boshernitzan}, admits factor frequencies. This proves Lem.~\ref{lem:boshernitzan-perm}.
\qed

\section{Discussion}
An open question is to find an example of a word $x$ (or prove one does not exist) that does not satisfy Boshernitzan's condition, but $\mathcal{A}(x)$ generates uniquely ergodic shifts for every $\mathcal{A}$. 

An obvious direction for future work is to generalise the following beyond the case of primitive morphic words: if the factor frequencies of $x$ are given by \emph{effective closed-form} expressions, then so are those of $\mathcal{A}(x)$.
Concretely, consider the ability to decide whether $f_x(u) = r$ where $r$ is a given constant: our current proof of Thm.~\ref{thm:boshernitzan-preserve} via Lem.~\ref{lem:ur-ue-compfreq} does not imply that $\mathcal{A}(x)$ inherits this property. 

The main technical obstacle to the above is obtaining an effective version of Lem.~\ref{lemma:coboundary}, i.e., the computation of \emph{minimal cobounding maps} on the shift $X$ defined by the input word. This is closely related to the behaviour of \emph{return groups}. Recall that $\mathcal{R}_X(u)$ denotes the set of return words to a factor $u \in \mathcal{L}(X)$. Let $G$ be a group, and $\varphi$ be a morphism from $\Sigma^*$ into $G$. The return group of $u$ (with respect to $G, \varphi$) is generated by the return words to $u$, and is given by $\langle \varphi(\mathcal{R}_X(u)) \rangle$. If $u'$ is a prefix of $u$, then the return group of $u$ is a subgroup of that of $u'$ \cite[Lem.~3]{stability-return}. We say that the return groups \emph{stabilise} at $u$ if for every factor $u''$ such that $u$ is a prefix of $u''$, the return groups of $u$ and $u''$ are the same. As \cite[Prop.~7.8]{berthe2024density} observes, determining minimal cobounding maps is equivalent to determining when return groups stabilise.

Our effective results for primitive morphic words indicate that the stabilisation of return groups can be determined in this case, and it has indeed been proven \cite[Prop.~27]{stability-return}. Nevertheless, this remains open in the general case, to the best of our knowledge. A possible approach is to generalise the derivation-based techniques of Durand applied in both \cite{stability-return} and our paper.

Another potential approach leverages recent advances in word combinatorics, in particular the study of so-called \emph{suffix-connected} words and their generalisations \cite{hermansuffixconnected}. An important property of suffix-connected words (which include Sturmian words and Arnoux-Rauzy words) is that all factors are \emph{stable} \cite[Cor.~1.2]{hermansuffixconnected}, i.e., for any factor $u$, the return group (with respect to $G, \varphi$) is $G$. While this property is ideal to prove an effective version of Lem.~\ref{lem:boshernitzan-perm}, it remains difficult to determine when an automaton preserves suffix-connectedness, and hence obtain a preservation result for an entire cascade of automata. As a positive case in point, we give the following application.

\begin{proposition}
\label{lem:sturmian-perm}
    Let $\mathcal{A} = (Q', q_\mathsf{init}, \Sigma, \delta)$ be a permutation automaton, and let $x \in \Sigma^\omega$ be a Sturmian word. Let $G$ be the group of permutations induced by the transitions, and let $Q$ be the orbit of $q_\mathsf{init}$ under $G$. The run of $\mathcal{A}$ on $x$ visits each state $q \in Q$ with frequency $1/|Q|$.
\end{proposition}
\begin{proof}
    Let $X$ be the minimal shift defined by $x$. Since the word $x$ is stable by virtue of being Sturmian and hence suffix-connected, we have by \cite[Prop.~7.8]{berthe2024density} that the trivial cobounding map on $X$ is minimal. This implies that the shift $Y$ corresponding to the entire skew product $G \rtimes X$ is minimal. Furthermore, since $x$ is a Sturmian word, it satisfies Boshernitzan's condition (Rmk.~\ref{remark:Boshernitzan}), and by Lem.~\ref{lem:boshernitzan-perm}, the shift $Y$ also satisfies Boshernitzan's condition. From the proof of Lem.~\ref{lem:boshernitzan-perm}, we have that the unique invariant measure on $Y$ is given by $\mu_Y([(g, u)]_Y) = \frac{1}{|G|} \mu_X([u]_X)$.

    Consider $q \in Q$. The set $H_q = \{g \mid q_\mathsf{init} \cdot g = q\}$ is a coset of the subgroup $H = \{h \mid q_\mathsf{init} \cdot h = q_\mathsf{init}\}$, and hence $|Q|\cdot |H| = |G|$. Now by Oxtoby's theorem, $f_{\mathcal{A}(x)}(q, a) = \sum_{g \in H_q} \mu_Y([(g, a)]_Y)$, which by the definition of $\mu_Y$ and the above observation, simplifies to $\frac{1}{|Q|}\mu_X([a]_X) = \frac{1}{|Q|}f_x(a)$. The frequency with which a run visits a state $q \in Q$ is then simply $\frac{1}{|Q|}\sum_a f_x(a) = \frac{1}{|Q|}$.
\end{proof}

We conclude by recording an interesting corollary\footnote{We acknowledge that this fact was known to Olivier Carton and Vincent Delecroix.} of Thm.~\ref{thm:boshernitzan-preserve}, Rmk.~\ref{remark:Boshernitzan}, and the proof of Prop.~\ref{lem:sturmian-perm}.

\begin{corollary}
    \label{cor:Sturmian}
    Let $\Sigma = \{0, 1\}$, let $\mathcal{A} = (Q, q_\mathsf{init}, \Sigma, \Gamma, \delta^o, \delta^t)$ be a transducer, and let $x \in \Sigma^\omega$ be a Sturmian word whose factor frequencies have effective closed form expressions. If  $\mathcal{A}$ is a permutation transducer, then the factor frequencies in $\mathcal{A}(x)$ have effective closed form expressions.
\end{corollary}

\bibliographystyle{plain}
\bibliography{main}

\end{document}